\documentclass[11pt,a4paper]{article}
\pdfoutput=1

\usepackage[margin=1in]{geometry}
\usepackage{amsmath,amssymb,amsthm,booktabs,color,doi,enumitem,graphicx,latexsym,url,xcolor}
\usepackage[numbers,sort&compress]{natbib}
\usepackage{microtype,hyperref}
\usepackage{thmtools,thm-restate}

\definecolor{citecolor}{HTML}{0000C0}
\definecolor{urlcolor}{HTML}{000080}

\hypersetup{
    colorlinks=true,
    linkcolor=black,
    citecolor=citecolor,
    filecolor=black,
    urlcolor=black,
}

\newtheorem{theorem}{Theorem}
\newtheorem{lemma}{Lemma}
\newtheorem{corollary}{Corollary}
\newtheorem{definition}{Definition}

\theoremstyle{remark}

\newcommand{\namedref}[2]{\hyperref[#2]{#1~\ref*{#2}}}
\newcommand{\sectionref}[1]{\namedref{Section}{#1}}

\newcommand{\theoremref}[1]{\namedref{Theorem}{#1}}
\newcommand{\corollaryref}[1]{\namedref{Corollary}{#1}}

\newcommand{\figureref}[1]{\namedref{Figure}{#1}}
\newcommand{\lemmaref}[1]{\namedref{Lemma}{#1}}

\newcommand{\defref}[1]{\namedref{Definition}{#1}}

\renewcommand{\vec}[1]{\mathbf{#1}}

\newcommand{\bN}{\mathbb{N}}

\newcommand{\sgo}{\ensuremath{\operatorname{GO}}}
\newcommand{\sfire}{\ensuremath{\operatorname{FIRE}}}

\DeclareMathOperator{\polylog}{polylog}

\DeclareMathOperator{\N}{\mathbb N}

\newenvironment{mycover}
               {\list{}{\listparindent 0pt
                        \itemindent    \listparindent
                        \leftmargin    0pt
                        \rightmargin   0pt
                        \parsep        0pt}%
                \raggedright
                \item\relax}
               {\endlist}

\begin{document}

\begin{mycover}
  
{\LARGE \textbf{Near-Optimal Self-Stabilising Counting and \\Firing Squads}\par}

\bigskip
\bigskip

\bigskip 
\textbf{Christoph Lenzen}\, $\cdot$\, \href{mailto:clenzen@mpi-inf.mpg.de}{\nolinkurl{clenzen@mpi-inf.mpg.de}}

\smallskip
{\small Department of Algorithms and Complexity, \\
Max Planck Institute for Informatics, \\
Saarland Informatics Campus \par}

\bigskip
\textbf{Joel Rybicki}\, $\cdot$\, \href{mailto:joel.rybicki@helsinki.fi}{\nolinkurl{joel.rybicki@helsinki.fi}}

\smallskip
{\small Department of Biosciences, University of Helsinki\footnote[1]{Current affiliation of JR.} \par}
\medskip

{\small Helsinki Institute for Information Technology HIIT, \\
Department of Computer Science, Aalto University\par}

\end{mycover}

\bigskip
\paragraph{Abstract.}

Consider a fully-connected synchronous distributed system consisting of $n$ nodes, where up to $f$ nodes may be faulty and every node starts in an arbitrary initial state. In the \emph{synchronous $C$-counting} problem, all nodes need to eventually agree on a counter that is increased by one modulo $C$ in each round for given $C>1$. In the \emph{self-stabilising firing squad} problem, the task is to eventually guarantee that all non-faulty nodes have simultaneous responses to external inputs: if a subset of the correct nodes receive an external ``go'' signal as input, then all correct nodes should agree on a round (in the not-too-distant future) in which to jointly output a ``fire'' signal. Moreover, no node should  generate a ``fire'' signal without some correct node having previously received a ``go'' signal as input.

We present a framework reducing both tasks to binary consensus at very small cost. For example, we obtain a deterministic algorithm for self-stabilising Byzantine firing squads with optimal resilience $f<n/3$, asymptotically optimal stabilisation and response time $O(f)$, and message size $O(\log f)$. As our framework does not restrict the type of consensus routines used, we also obtain efficient randomised solutions, and it is straightforward to adapt our framework for other types of permanent faults.

\thispagestyle{empty}
\setcounter{page}{0}
\newpage

\section{Introduction}\label{sec:intro}

The design of distributed systems faces several unique issues related to redundancy and fault-tolerance, timing and synchrony, and the efficient use of communication as a resource~\cite{lynch96book}. In this work, we give near-optimal solutions to two fundamental distributed synchronisation and coordination tasks: the synchronous counting and the firing squad problems. For both tasks, we devise fast self-stabilising algorithms~\cite{dolev00self-stabilization} that are not only communication-efficient, but also tolerate the optimal number of permanently faulty nodes. That is, our algorithms efficiently recover from transient failures that may arbitrarily corrupt the state of the distributed system \emph{and} permanently damage a large number of the nodes.

\subsection{Synchronous counting and firing squads}

We assume a synchronous message-passing model of distributed computation. The distributed system consists of a fully-connected network of $n$ nodes, where up to $f$ of the nodes may be faulty and the initial state of the system is arbitrary. To model the behaviour of faulty nodes, we consider three typical classes of permanent faults:
\begin{itemize}[noitemsep]
 \item crash (the faulty node stops sending information),
 \item omission (some or all of the messages sent by the faulty node are lost), and 
 \item Byzantine faults (the faulty node exhibits arbitrary misbehaviour).
\end{itemize}

Note that even though the communication proceeds in a synchronous fashion, the nodes may have different notions of current time due to the arbitrary initial states. However, many typical distributed protocols assume that the system has either been properly initialised or that the nodes should collectively agree on the rounds in which to perform certain actions. Thus, we are essentially faced with the task of having to agree on a common time in a manner that is both self-stabilising and tolerates permanently faulty behaviour from some of the nodes. To address this issue, we study the synchronous counting and firing squad problems, which are among the most fundamental challenges in fault-tolerant distributed systems. 

In the \emph{synchronous counting} problem, all the nodes receive well-separated synchronous clock pulses that designate the start of a new round. The received clock pulses are anonymous, and hence, all correct nodes should eventually stabilise and agree on a round counter that increases consistently by one modulo $C$. The problem is also known as \emph{digital clock synchronisation}, as all non-faulty nodes essentially have to agree on a shared logical clock. A stabilising execution of such a protocol for $n=4$, $f=1$, and $C=3$ is given below:

\begin{center}
 \includegraphics[page=1,scale=1.0]{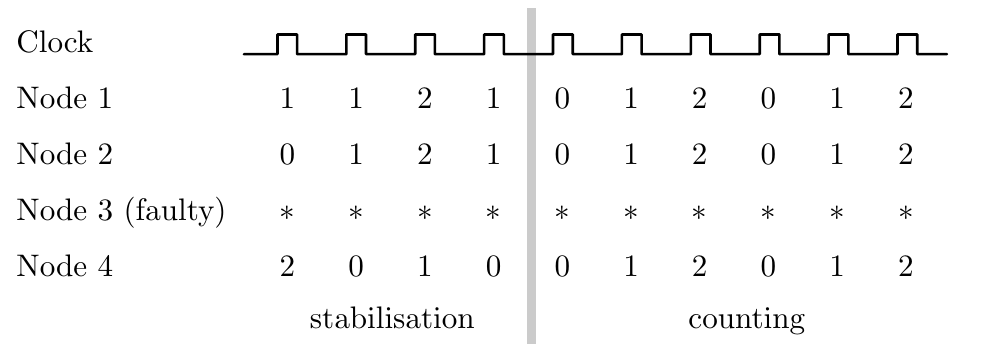}
\end{center}

 In the \emph{self-stabilising firing squad problem}, the task is to have all correct nodes eventually stabilise and respond to an external input simultaneously. That is, once stabilised, when a sufficiently large (depending on the type of permanent faults) subset of the correct nodes receive an external ``go'' signal, then all correct nodes should eventually generate a local ``fire'' event on the same round. The time taken to react to the ``go'' signal is called the response time. Note that before stabilisation the nodes may generate spurious firing signals, but after stabilisation no correct node should generate a ``fire'' event without some correct node having previously received a ``go'' signal as input. An execution of such a protocol with $n=4$, $f=1$, and response time $R=5$ is illustrated below:

 \begin{center}
 \includegraphics[page=2,scale=1.0]{squad-figures.pdf}
\end{center}

 A firing squad protocol can be used, for example, to agree in a self-stabilising manner on when to initiate a new instance of a non-self-stabilising distributed protocol, as response to internal or external ``go'' inputs.

\subsection{Connections to fault-tolerant consensus}

Reaching agreement is perhaps the most intrinsic problem in fault-tolerant distributed computing. It is known that both the synchronous counting~\cite{dolev15survey} and the self-stabilising firing squad problem~\cite{dolev12optimal} are closely connected to the well-studied consensus problem~\cite{pease80reaching,lamport82byzantine}, where each node is given an input bit and the task is to agree on a common output bit such that if every non-faulty node received the same value as input, then this value must also be the output value. Indeed, the connection is obvious on an intuitive level, as in each task the goal is to agree on a common decision (that is, the output bit, clock value, or whether to generate a firing event).

However, the key difference between the problems lies in self-stabilisation. Typically, the consensus problem is considered in a non-self-stabilising setting with only permanent faults (e.g.\ $f < n/3$ nodes with arbitrary behaviour), whereas synchronous counting copes with both transient and permanent faults. In fact, it is easy to see that synchronous counting is trivial in a non-self-stabilising setting: if all nodes are initialised with the same clock value, then they can simply locally increment their counters each round without any communication. Furthermore, in a properly initialised system, one can reduce the firing squad problem to repeatedly calling a consensus routine~\cite{burns85firing}. 

Interestingly, imposing the requirement of self-stabilisation -- convergence to correct system behavior from arbitrary initial states -- reverses the roles. Solving either the synchronous counting or firing squad problem in a self-stabilising manner also yields a solution to binary consensus, but the converse is not true. In fact, in order to internally or externally trigger a consistent execution of a consensus protocol (or any other non-self-stabilising protocol, for that matter), one first needs a self-stabilising synchronous counting or firing squad algorithm, respectively!

In light of this, the self-stabilising variants of both problems are important generalisations of consensus. While considerable research has been dedicated to both tasks~\cite{dolev12optimal,dolev15survey,lenzen15towards,lenzen15efficient,dolev07actions,ben-or08fast,dolev04clock-synchronization,dolev16synthesis}, our understanding is significantly less developed than for the extensively studied consensus problem. Moreover, it is worth noting that all existing algorithms utilise consensus subroutines~\cite{dolev07actions,lenzen15towards,lenzen15efficient} or shared coins~\cite{ben-or08fast}, the latter of which essentially solves consensus as well. Given that both tasks are at least as hard as consensus~\cite{dolev15survey}, this seems to be a natural approach. However, it raises the question how much of an overhead must be incurred by such a reduction. In this paper, we subsume and improve upon previous results by providing a generic reduction of synchronous counting and self-stabilising firing squad to binary consensus that incurs very small overheads.

\subsection{Contributions}

We develop a framework for efficiently transforming \emph{non-self-stabilising} consensus algorithms into \emph{self-stabilising} algorithms for synchronous counting and firing squad problems. In particular, the resulting self-stabilising algorithms have the same resilience as the original consensus algorithms, that is, the resulting algorithms tolerate the same number and type of permanent faults as the original consensus algorithm (e.g.\ crash, omission, or Byzantine faults). 

The construction we give incurs a small overhead compared to time and bit complexity of the consensus routines: the stabilisation time and message size are, up to constant factors, given as the sum of the cost of the consensus routine for $f$ faults and recursively applying our scheme to $f'<f/2$ faults. Finally, our construction can be used in conjunction with both deterministic and randomised consensus algorithms. Consequently, we also obtain algorithms for probabilistic variants of the synchronous counting and firing squad problems.

Our novel framework enables us to address several open problems related to self-stabilising firing squads and synchronous counting. We now give a brief summary of the open problems we solve and the new results obtained using our framework.

\paragraph{Self-stabilising firing squads.}
In the case of self-stabilising firing squads, Dolev et al.~\cite{dolev12optimal} posed the following two open problems:
\begin{enumerate}[noitemsep]
 \item Are there solutions that tolerate either omission or Byzantine (i.e., arbitrary) faults?
 \item Are there algorithms using $o(n)$-bit messages only?
\end{enumerate}
We answer both questions in the affirmative by giving algorithms that achieve both properties \emph{simultaneously}. Concretely, our framework implies a deterministic solution for the self-stabilising Byzantine firing squad problem that 
\begin{itemize}[noitemsep]
 \item tolerates the optimal number of $f < n/3$ Byzantine faulty nodes, 
 \item uses messages of $O(\log f)$ bits, and
 \item is guaranteed to stabilise and respond to inputs in linear-in-$f$ communication rounds.
\end{itemize}
Thus, compared to prior state-of-the-art solutions~\cite{dolev12optimal}, our algorithm tolerates a much stronger form of faulty behaviour and uses exponentially smaller messages, yet retains asymptotically optimal stabilisation and response time. We also obtain algorithms that tolerate $f < n/2$ omission failures and $f < n$ crash failures while retaining a small message size of $O(\log f)$ bits.

\paragraph{Synchronous counting.}
We attain novel algorithms for synchronous counting, which is also known as self-stabilising Byzantine fault-tolerant \emph{digital clock synchronisation}~\cite{dolev04clock-synchronization,hoch06digital,ben-or08fast}. Our new algorithms resolve questions left open by our own prior work~\cite{lenzen15towards}, namely, whether there exist
\begin{enumerate}[noitemsep]
 \item deterministic linear-time algorithms with optimal resilience and message size $o(\log^2 f)$, or 
 \item randomised sublinear-time algorithms with small bit complexity.
\end{enumerate}
Again, we answer both questions positively using our framework developed in this paper. For the first question, we give linear-time deterministic algorithms that have message size $O(\log f)$ bits. For the second question, we show that our framework can utilise efficient randomised consensus algorithms to obtain probabilistic variants of the synchronous counting and firing squad problems. For example, the result of King and Saia~\cite{king11breaking} implies algorithms that stabilise with high probability in $\polylog n$ rounds and use message size $\polylog n$, assuming private communication links and an adaptive Byzantine adversary corrupting $f<n/(3+\varepsilon)$ nodes for an arbitrarily small constant $\varepsilon>0$.

\subsection{Related work}\label{ssec:related}

In this section, we overview prior work on the synchronous counting and firing squad problems. By now it has been established that both problems~\cite{dolev12optimal,dolev15survey} are closely connected to the well-studied (non-self-stabilising) consensus~\cite{pease80reaching,lamport82byzantine}. As there exists a vast body of literature on synchronous consensus, we refer the interested reader to e.g.\ the survey by Raynal~\cite{raynal10survey}. We note that self-stabilising variants of consensus have been studied~\cite{doerr11stabilizing,sdolev10self-stabilization,angluin06stabilizing,daliot06self-stabilizing} but in different models of computation and/or for different types of failures than what we consider in this work. 

\paragraph{Synchronous counting and digital clock synchronisation.}
In the past two decades, there has been increased interest in combining self-stabilisation with Byzantine fault-tolerance. One reason is that algorithms in this fault model are very attractive in terms of designing highly-resilient hardware~\cite{dolev15survey}. A substantial amount of work on synchronous counting has been carried out~\cite{dolev04clock-synchronization,hoch06digital,dolev07actions,ben-or08fast,dolev16synthesis,lenzen16extended}, comprising both positive and negative results.

In terms of lower bounds, many impossibility results for consensus~\cite{pease80reaching,fischer82lower,dolev82byzantine,dolev85bounds} also directly apply to synchronous counting, as synchronous counting solves binary consensus~\cite{dolev15survey,dolev16synthesis}. In particular, no algorithm can tolerate more than $f<n/3$ Byzantine faulty nodes~\cite{pease80reaching} (unless cryptographic assumptions are made) and any deterministic algorithm needs at least $f+1$ rounds to stabilise~\cite{fischer82lower}.

In a seminal work, Dolev and Welch~\cite{dolev04clock-synchronization} showed that the task can be solved in a self-stabilising manner in the presence of (the optimal number of) $f < n/3$ Byzantine faults using randomisation; see also~\cite[Ch.~6]{dolev00self-stabilization}. While this algorithm can be implemented using only constant-size messages, the expected stabilisation time is exponential. Later, Ben-Or et al.~\cite{ben-or08fast} showed that it is possible to obtain optimally-resilient solutions that stabilise in expected constant time. However, their algorithm relies on shared coins, which are costly to implement and assume private communication channels. 

In addition to the lower bound results, there also exist deterministic algorithms for the synchronous counting problem~\cite{hoch06digital,dolev07actions,dolev16synthesis,lenzen16extended}. Many of these algorithms utilise consensus routines~\cite{hoch06digital,dolev07actions,lenzen16extended}, but obtaining fast and communication-efficient solutions with optimal resilience has been a challenge. For example, Dolev and Hoch~\cite{dolev07actions} apply a pipelining technique, where $\Omega(f)$ consensus instances are run in parallel. While this approach attains optimal resilience and linear stabilisation time in $f$, the large number of parallel consensus instances necessitates large messages. 

In order to achieve better communication and state complexity, the use of computational algorithm design and synthesis techniques have also been investigated~\cite{dolev16synthesis,bloem16synthesis}. While this line of research has produced novel optimal and computer-verified algorithms, so far the techniques have not scaled beyond $f=1$ faulty node due to the inherent combinatorial explosion in the search space of potential algorithms.

Recently, we gave recursive constructions that achieve linear stabilisation time using only polylogarithmic message size and state bits per node~\cite{lenzen15towards,lenzen15efficient}; see also the extended and revised version~\cite{lenzen16extended}. However, our previous constructions relied on specific (deterministic) consensus routines and their properties in a relatively ad hoc manner. In contrast, our new framework presented here lends itself to any (possibly randomised) synchronous consensus routine and improves the best known upper bound on the message size to $O(\log f)$ bits. Currently, it is unknown whether it is possible to deterministically achieve message size $o(\log f)$.

\paragraph{Firing squads.}

In the original formulation of the firing squad synchronisation problem, the system consists of an $n$-length path consisting of finite state machines (whose number of states is independent of $n$) and the goal is to have all machines switch to the same ``fire'' state simultaneously after one node receives a ``go'' signal. This formulation of the problem has been attributed to John Myhill and Edward Moore and has subsequently been studied in various settings; see e.g.~\cite{nishitani81firing} for survey of early work related to the problem.

In the distributed computing community, the firing squad problem has been studied in fully-conneted networks in the presence of faulty nodes. Similarly to synchronous counting, the firing squad problem is closely connected to Byzantine agreement and simultaneous consensus~\cite{burns85firing,coan89firing,dwork90knowledge,coan91simultaneity,dolev12optimal}. Both Burns and Lynch~\cite{burns85firing} and Coan et al.~\cite{coan89firing} studied the firing squad problem in the context of Byzantine failures. Burns and Lynch~\cite{burns85firing} considered both permissive and strict variants of the problem (i.e., whether faulty nodes can trigger a firing event or not) and showed that both can be solved using Byzantine consensus algorithms with only a relatively small additional overhead in the number of communication rounds and total number of bits communicated. On the other hand, Coan et al.~\cite{coan89firing} gave \emph{authenticated} firing squad algorithms for various Byzantine fault models. Coan and Dwork~\cite{coan91simultaneity} gave time lower bounds of $f+1$ rounds for deterministic and randomised algorithms solving the firing squad problem in the crash fault model. 

However, neither the solutions of Burns and Lynch~\cite{burns85firing} or Coan et al.~\cite{coan89firing} are self-stabilising or use small messages. Almost two decades later, Dolev et al.~\cite{dolev12optimal} gave the first self-stabilising algorithm for the firing squad problem. In particular, their solution has optimal stabilisation time and response time depending on the fault pattern. However, their algorithm tolerates only crash faults and uses messages of size $\Theta(n \log n)$ bits. In this work, we improve on this result by achieving Byzantine fault-tolerance using messages of $O(\log n)$ bits. 

\subsection{Outline of the paper}

The article is structured as follows. For the first part of the paper, we confine the presentation to Byzantine faults. In the second part, we discuss how to extend our results in two ways: first, we consider the randomised setting, where sublinear time algorithms are possible, and secondly, other fault models that allow a larger number of faulty nodes.

We start with \sectionref{sec:preliminaries}, where we give formal definitions related to the model of computation, synchronous counting, and firing squads in the Byzantine setting. In the sections following this, we show our main result in a top-down fashion as illustrated in \figureref{fig:dependency-overview}. We introduce a series of new problems and give reductions between them: 
\begin{itemize}[noitemsep]
 \item \sectionref{sec:firing-squad} shows how to obtain synchronous counting and firing squad algorithms that rely on binary consensus routines and \emph{strong pulsers},
 \item \sectionref{sec:strong-pulsers} devises strong pulsers with the help of \emph{weak pulsers} and multivalued consensus, 
 \item \sectionref{sec:weak-pulsers} constructs weak pulsers using \emph{silent consensus} and less resilient strong pulsers.
\end{itemize}
\sectionref{sec:main} combines the results of \sectionref{sec:strong-pulsers} and \sectionref{sec:weak-pulsers} to obtain a recursive construction for strong pulsers used by the algorithms given in \sectionref{sec:firing-squad}. Finally, to demonstrate the flexibility and generality of our approach, we discuss how to extend our results to randomised consensus routines in \sectionref{sec:rand}, and cover deterministic solutions under crash and omission faults in \sectionref{sec:other}.

\begin{figure}
\begin{center}
 \includegraphics[page=5,scale=1.0]{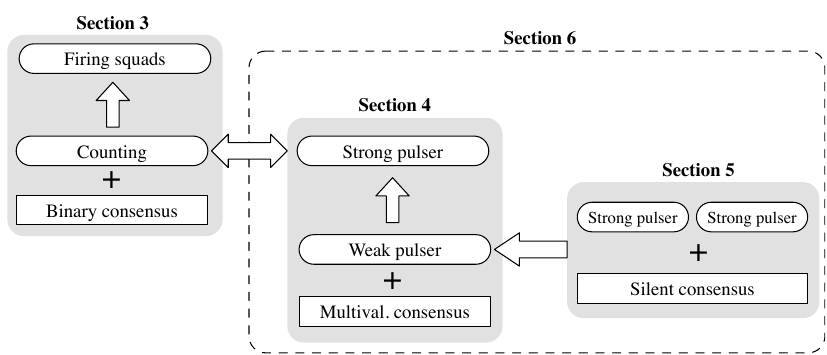}
\end{center}
\caption{High-level overview of our construction and the structure of the paper. Rounded boxes denote algorithms that are both self-stabilising and Byzantine fault-tolerant, whereas rectangular boxes denote non-stabilising Byzantine fault-tolerant consensus routines.\label{fig:dependency-overview}}
\end{figure}

\section{Preliminaries}\label{sec:preliminaries}

In this section, we first fix some basic notation, then describe the model of computation, and finally give formal definitions of the synchronous counting, self-stabilising firing squad, and consensus problems.

\subsection{Notation}

We use $\bN = \{1, 2, \ldots \}$ to denote the set of positive integers and $\bN_0 = \{0\} \cup \bN$ to denote the set of all non-negative integers. For any $k \in \bN$, we write $[k] = \{0, \ldots, k-1\}$ to be the set of the first $k$ non-negative integers.

\subsection{Model of computation}

We consider a fully-connected synchronous network on node set $V$ consisting of $n=|V|$ processors. We assume there exists a subset of $F\subseteq V$ faulty nodes that is (at least initially) unknown to all nodes, where the upper bound $f$ on the size $|F| \le f$ is known to the nodes. We say that nodes in $V \setminus F$ are correct and nodes in $F$ are faulty. 

All correct nodes in the system will follow a given algorithm $\vec A$ that is the same for all the nodes in the system. The execution proceeds in synchronous rounds, where in each round $t \in \bN$ the nodes take the following actions in lock-step:
\begin{enumerate}[noitemsep]
 \item perform local computations, 
 \item send messages to other nodes, and 
 \item receive messages from other nodes. 
\end{enumerate}
We assume that nodes have unique identifiers from $\{1,\ldots,n\}$ and can identify the sender of incoming messages.

We say that an algorithm $\vec A$ has message size $M(\vec A)$ if no correct node sends more than $M(\vec A)$ bits to any other node during a single round.

The local computations of a node $v$ determine the decision which messages to send to other nodes and what is the new state of the node $v$. As we are interested in self-stabilising algorithms, the initial state of a node is arbitrary; this is equivalent to assuming that transient faults have arbitrarily corrupted the state of each node, but the transient faults have ceased by the beginning of the first round. 

As mentioned above, we allow for additional (possibly permanent) \emph{Byzantine} faults. A Byzantine faulty node $v\in F$ may deviate from the algorithm arbitrarily, i.e., send arbitrary messages in each round. In particular, a Byzantine faulty node can send different messages to each correct node in the system, even if the algorithm specifies otherwise. Since we consider deterministic algorithms, the meaning of ``arbitrary'' in this context is that the algorithm must succeed for any possible choice of behavior of the faulty nodes. We require that $f=|F|<n/3$, as otherwise none of the problems we consider can be solved due to the impossibility of consensus under $f\geq n/3$ Byzantine faults~\cite{pease80reaching}.

\subsection{Synchronous counting}

In the \emph{synchronous $C$-counting} problem, the task is to have each node $v \in V$ output a counter value $c(v,t) \in [C]$ on each round $t\in \bN$ in a consistent manner. We say that an execution of an algorithm \emph{stabilises in round $t$} if and only if all $t\leq t'\in \bN$ and $v,w\in V\setminus F$ satisfy
\begin{enumerate}[label=SC\arabic*.,noitemsep]
 \item {\bf Agreement:} $c(v,t')=c(w,t')$ and 
 \item {\bf Consistency:} $c(v,t'+1)=c(v,t')+1\bmod C$.
\end{enumerate}
We say that $\vec A$ is an $f$-resilient $C$-counting algorithm that stabilises in time $t$ if all executions with at most $f$ faulty nodes stabilise by round $t$. The stabilisation time $T(\vec A)$ of $\vec A$ is the maximum such $t$ over all executions.

\subsection{Self-stabilising firing squad}\label{ssec:firing}

In the self-stabilising Byzantine firing squad problem, in each round $t\in \bN$, each node $v\in V$ receives an external input $\sgo(v,t)\in \{0,1\}$. Moreover, the algorithm determines an output $\sfire(v,t)\in \{0,1\}$ at each node $v\in V$ in each round $t\in \bN$. We say that an execution of an algorithm \emph{stabilises in round $t\in \bN$} if the following three properties hold:
\begin{enumerate}[label=FS\arabic*.,noitemsep]
  \item {\bf Agreement:} $\sfire(v,t')=\sfire(w,t')$ for all $v,w\in V\setminus F$ and $t\leq t'\in \bN$.
  \item {\bf Safety:} If $\sfire(v,t_F)=1$ for $v\in V\setminus F$ and $t\leq t_F\in \bN$, then there is $t_F\geq t_G\in \bN$ s.t.
\begin{enumerate}[label=(\roman*)]
 \item $\sgo(w,t_G)=1$ for some $w\in V\setminus F$ and 
 \item $\sfire(v,t')=0$ for all $t'\in \{t_G+1,\ldots,t_F-1\}$.
\end{enumerate}
  \item {\bf Liveness:} If $\sgo(v,t_G)=1$ for at least $f+1$ nodes $v\in V\setminus F$ and $t\leq t_G\in \bN$, then $\sfire(v,t_F)=1$ for all nodes $v\in V\setminus F$ and some $t_G< t_F\in \bN$.
\end{enumerate}
Note that the liveness condition requires $f+1$ correct nodes to observe a $\sgo$ input, as otherwise it would be impossible to guarantee that a correct node observed a $\sgo$ input when firing; this corresponds to the definition of a strict Byzantine firing squad~\cite{burns85firing}. We say that an execution stabilised by round $t$ has \emph{response time $R$ from round $t$} if
\begin{enumerate}[label=(\roman*),noitemsep]
 \item when firing is required in response to (sufficiently many) $\sgo$ inputs of $1$ in round $t_G\geq t$, this happens no later than round $t_G+R$, and 
 \item when the squad fires in round $t_F\geq t$, there was sufficient support (in terms of $\sgo$ inputs of $1$) justifying this in a round $t_G$ with $t_F > t_G \geq t_F - R$. 
\end{enumerate}
Finally, we say that an algorithm $\vec F$ is an $f$-resilient firing squad algorithm with stabilisation time $T(\vec F)$ and response time $R(\vec F)$ if in any execution of the system with at most $f$ faulty nodes there is a round $t\leq T(\vec F)$ such that the algorithm stabilised and has response time at most $R(\vec F)$ from round $t$.

We remark that under Byzantine faults, previous non-stabilising algorithms~\cite{burns85firing} have considered the case where the input signals (from different nodes) do not need to be received on the same round, but they can be spread out over several rounds. In the self-stabilising setting, we can easily cover the case where $f+1$ input signals are received within a time window of $\Delta$ rounds as follows: instead of relying on the input $\sgo$ signals as-is, we can use an auxiliary variable $\sgo'(v,t)$ as input to our algorithms, where $\sgo'(v,t)=1$ iff there is a round $t'\in \{t-\Delta+1,\ldots,t\}$ with $\sgo(v,t')=1$.

\subsection{Consensus}

Let us conclude this section by definining the \emph{multivalued consensus problem}. Unlike the synchronous counting and self-stabilising firing squad problems, the standard definition of consensus does \emph{not} require self-stabilisation: we assume that all nodes start from a fixed starting state and the algorithm terminates in finitely many communication rounds.

In the multivalued consensus problem for $L>1$ values, each node $v \in V$ receives an input value $x(v) \in [L]$ and the task is to have all correct nodes output the same value $y \in [L]$. We say that an algorithm $\vec C$ is an $f$-resilient $T(\vec C)$-round consensus algorithm if the following conditions hold when there are at most $f$ faulty nodes:
\begin{enumerate}[label=C\arabic*.,noitemsep]
 \item {\bf Termination:} Each $v \in V \setminus F$ decides on an output $y(v) \in [L]$ by the end of round $T(\vec C)$.
 \item {\bf Agreement:} For all $v, w \in V \setminus F$, it holds that $y(v) = y(w)$.
 \item {\bf Validity:} If there exists $x \in [L]$ such that for all $v \in V \setminus F$ it holds that $x(v)=x$, then each $v \in V \setminus F$ outputs the value $y(v) = x$.
\end{enumerate}
We remark that one may ask for stronger validity conditions, but for our purposes this condition is sufficient. The \emph{binary consensus} problem is the special case of $L=2$ of the above multivalued consensus problem. In the case of binary consensus, the stated validity condition is equivalent to requiring that if $v \in V \setminus F$ outputs $y(v) = x \in \{0,1\}$, then some $w \in V \setminus F$ has input value $x(w)=x$.

Later, we utilise the fact that multivalued consensus can be reduced to binary consensus with only a small overhead in time. In~\cite{lenzen13time}, it is shown how to do this with $1$-bit messages and an additive overhead of $O(\log L)$ rounds, preserving resilience. 

\begin{theorem}[\cite{lenzen13time}]\label{thm:multi-valued}
Let $L > 1$. Given an $f$-resilient binary consensus algorithm $\vec C$, we can solve $L$-value consensus in $O(\log L) + T(\vec C)$ rounds using $M(\vec C)$-bit messages while tolerating $f$ faults.
\end{theorem}

\section{Synchronous counting and firing squads}\label{sec:firing-squad}

In this section, we give a firing squad algorithm with asymptotically optimal stabilisation and response times. The algorithm relies on two auxiliary routines: a so-called strong pulser and a consensus algorithm. We start with a discussion on strong pulsers.

\subsection{Strong pulsers and counting}

Our approach to the firing squad problem is to solve it by repeated consensus, which in turn is controlled by a joint round counter. To minimise message size, however, we will not communicate counter values directly. Instead we make use of what we call a \emph{strong pulser}.
\begin{definition}[Strong pulser]
An algorithm $\vec P$ is an $f$-resilient strong $\Psi$-pulser that stabilises in $T(\vec P)$ rounds if it satisfies the following conditions in the presence of at most $f$ faulty nodes. Each node $v \in V$ produces an output bit $p(v,t) \in \{0,1\}$ on each round $t \in \bN$. We say that $v$ generates a pulse in round $t$ if $p(v,t)=1$ holds. We require that there is a round $t_0 \leq T(\vec P)$ such that:
\begin{enumerate}[label=S\arabic*.,noitemsep]
 \item For any $v\in V\setminus F$ and round $t=t_0+k\Psi$, where $k \in \bN_0$, it holds that $p(v,t)=1$.
 \item For any $v\in V\setminus F$ and round $t\geq t_0$ satisfying $t\neq t_0+k\Psi$ for $k \in \bN_0$, we have $p(v,t)=0$.
\end{enumerate}
\end{definition}
\begin{figure}
 \begin{center}
 \includegraphics[page=3,scale=1.0]{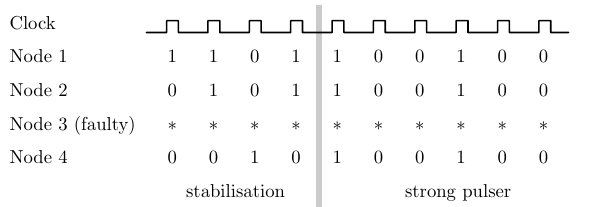}
\end{center}
\caption{An example execution of a strong 3-pulser on $n=4$ nodes with $f=1$ faulty node.\label{fig:strong-pulser}}
\end{figure}
Put otherwise, a strong $\Psi$-pulser consistently generates pulses at all non-faulty nodes exactly every $\Psi$ rounds. \figureref{fig:strong-pulser} illustrates an execution of a strong pulser with $\Psi = 3$. It is straightforward to see that strong pulsers and synchronous counting are almost equivalent.

\begin{lemma}\label{lemma:strong_pulsers_equal_counters}
Let $C \in \bN$ and $\Psi \in \bN$. If $C$ divides $\Psi$, then a strong $\Psi$-pulser that stabilises in $T$ rounds implies a synchronous $C$-counter that stabilises in at most $T$ rounds. If $\Psi$ divides $C$, then a synchronous $C$-counter that stabilises in $T$ rounds implies a strong $\Psi$-pulser that stabilises in at most $T+\Psi-1$ rounds.
\end{lemma}
\begin{proof}
For the first claim, set $c(v,t)=0$ in any round $t$ for which $p(v,t)=1$ and $c(v,t)=c(v,t-1)+1\bmod C$ in all other rounds. For the second claim, set $p(v,t)=1$ in all rounds $t$ in which $c(v,t)\bmod \Psi = 0$ and $p(v,t)=0$ in all other rounds.
\end{proof}
Another way of interpreting this relation is to view a strong $\Psi$-pulser as a different encoding of the output of a $\Psi$-counter: since the system is synchronous, it suffices to communicate when the counter overflows to value $0$ and otherwise count locally. This saves bandwidth when communicating the state of the counter.

\subsection{Firing squads via pulsers and consensus}\label{ssec:counting-squad}

We now show how an $f$-resilient strong pulser and $f$-resilient binary consensus algorithm can be used to devise an $f$-resilient firing squad algorithm. As a strong pulser can be used to control repeated execution of a non-self-stabilising algorithm, it enables us to run consensus on whether a firing event should be triggered or not repeatedly. As the firing squad problem is at least as hard as consensus, this maintains asymptotically optimal round complexity. 

Recall that for the \emph{Byzantine} firing squad problem, we are interested in a liveness condition in which a firing event needs to be generated if at least $f+1$ non-faulty nodes $v \in V \setminus F$ recently saw $\sgo(v,t)=1$ on some round $t$. To this end, we have each node continuously inform all other nodes about its \sgo\ values (i.e.\ their received input signals). Whenever node $v \in V$ sees $f+1$ nodes $w \in V$ claim $\sgo(w,t)=1$, it will memorise this and use input $x(v)=1$ for the next consensus instance. Otherwise, it will use the input value $x(v)=0$; this ensures that at least one non-faulty node $w$ had $\sgo(w,t)=1$ recently in case $v$ uses input $x(v)=1$. The validity condition of the (arbitrary) $T(\vec C)$-round consensus routine $\vec C$ thus ensures both liveness and safety for the resulting firing squad algorithm. Apart from $\vec C$, the algorithm concurrently runs a strong $\Psi$-pulser $\vec P$ for some $\Psi> T(\vec C)$.

\paragraph{The firing squad algorithm.}
Given a strong $\Psi$-pulser algorithm $\vec P$ and a binary consensus algorithm $\vec C$, each node $v$ stores the following variables on every round $t$:
\begin{itemize}[noitemsep]
 \item $p(v,t) \in \{0,1\}$, the output variable of $\vec P$, 
 \item $x(v,t) \in \{0,1\}$ and $y(v,t) \in \{0,1\}$, the input and output variables of $\vec C$, and 
 \item $m(v,t) \in \{0,1\}$, an auxiliary variable used to memorise whether sufficiently many \sgo\ signals were received to warrant a firing event.
\end{itemize}
In the following algorithm, on each round $t \in \bN$ any (correct) node $v \in V$ will broadcast the value $\sgo(v,t)$ and receive the values $\sgo(v,w,t-1)$  sent by every $w \in V$ in the previous round. The algorithm consists of each node $v$ executing the following operations\footnote{For better readability, we allow for statements about what a node communicates appearing anywhere in the description. Note, however, that sending operations happen after local computation, i.e., only information sent in the previous rounds is available for computations.}
in each round $t\in \bN$:
\begin{enumerate}[noitemsep]
 \item Broadcast $\sgo(v,t)$.
 \item If received at least $f+1$ nodes $w \in V$ sent $\sgo(v,w,t-1)=1$, then set $x(v,t) = 1$ and $m(v,t) = 1$. Otherwise, set $x(v,t) = x(v,t-1)$ and $m(v,t) = m(v,t-1)$.
 \item If $p(v,t)=1$, start executing a new instance of $\vec C$ using the value $x(v,t)$ as input and set $m(v,t) = 0$ while aborting any previously running instance. More specifically, this entails the following:
 \begin{itemize}[noitemsep]
   \item Maintain a local round counter $r$, which is initialised to $1$ on round $t$ and increased by $1$ after each round.
   \item Maintain the local state variables related to the consensus routine $\vec C$.
   \item On each round, execute round $r$ of algorithm $\vec C$; if the state variables indicate that $\vec C$ terminated at $v$, then do nothing.
   \item On the round when $r$ would attain the value $T(\vec C)+1$, stop the simulation (indicating this, e.g., by setting $r(v) = \bot$) and locally output the value of $y(v)$ computed by the simulation of $\vec C$.
 \end{itemize}
 \item If $\vec C$ outputs $y(v,t)=1$ on round $t$, then output $\sfire(v,t)=1$ and set $x(v,t) = 0$. \\ Otherwise, set $\sfire(v,t)=0$.
 \item If $\vec C$ outputs $y(v,t)=0$ on round $t$ and $m(v,t)=0$, then set $x(v,t) = 0$.
\end{enumerate}

We now show that the above algorithm satisfies the properties required from a self-stabilising firing squad.

\begin{theorem}\label{thm:squad}
Suppose there exists an $f$-resilient strong $\Psi$-pulser $\vec P$ and a consensus algorithm $\vec C$, where $\Psi > T(\vec C)$. Then there exists an $f$-resilient firing squad algorithm $\vec F$ that 
\begin{itemize}[noitemsep]
 \item stabilises in time $T(\vec F) \le T(\vec P) + \Psi$,
 \item has response time $R(\vec F) \le \Psi + T(\vec C)$, and
 \item uses message of size $M(\vec F) \le M(\vec P) + M(\vec C) + 1$ bits.
\end{itemize}
\end{theorem}
\begin{proof}
Let $\vec F$ be the algorithm described above. We now argue that the algorithm satisfies the three properties given in \sectionref{ssec:firing}: (FS1) agreement, (FS2) safety, and (FS3) liveness. We will show that the algorithm has a response time bounded by $R = T(\vec C) + \Psi$.

(FS1) Denote by $t_0\leq T(\vec P)$ the round in which the execution of the strong $\Psi$-pulser $\vec P$ has stabilised and generated a pulse. That is, for rounds $t \geq t_0$ we have that $p(v,t)=1$ is equivalent to $t = t_0+k\Psi$ for some $k \in \bN_0$. This implies that the algorithm will correctly simulate instances of the consensus routine $\vec C$ and locally output its decision on rounds $r_k = t_0+T(\vec C)+k\Psi < t_{k+1}$ for $k\in \bN_0$. The agreement property of the firing squad thus follows from the agreement property of consensus for all rounds $t\geq t_0$, as $\sfire(v,t)=1$ if and only if $t = r_k$ and the simulation of $\vec C$ output the value $y(v,t)=1$ in Step 4.

(FS2) Concerning safety, suppose $v\in V\setminus F$ outputs $\sfire(v,t_F)=1$ in round $t_F \ge t_1+T(\vec C)$. By the above discussion and the validity property of consensus, this implies that there was some node $w \in V \setminus F$ that started a (successfully and completely simulated) instance of $\vec C$ with input $x(w,t_k)=1$ in round $t_k=t_F-T(\vec C)=t_0+k\Psi$ and that $t_F = r_k$ for some $k\in \bN$. Assume for contradiction that there are no $u\in V\setminus F$ and $t_G \in \{ t_{k-1}, \ldots, t_k - 1\}$ satisfying $\sgo(u,t_G)=1$. Then, $w$ does not set $x(w,t')$ or $m(w,t')$ to $1$ in rounds $t' \in \{ t_{k-1}+1, \ldots, t_k \}$ in Step 2. However, in round $t_{k-1} = t_F-T(\vec C)-\Psi=t_0+(k-1)\Psi$ node $w$ set $m(w,t_{k-1})=0$ (by Step 3) and thus $w$ sets $x(w,r_{k-1}) = 0$ later in round $r_{k-1}$ (by Steps 4 and 5), the round in which the previous instance of $\vec C$ locally output some value. This contradicts the fact that $x(w,t_k)=1$ is set in round $t_k$. Hence, there must be $u\in V\setminus F$ and $t_G \in \{ t_{k-1}, \ldots, t_k - 1 \}$ such that $\sgo(u,t_G)=1$.

Recall that the above claimed existence of $u\in V\setminus F$ and $t_G$ such that $\sgo(u,t_G)=1$ is necessary for the safety condition to hold, but not sufficient. 
It is also required that $\sfire(v,t')=0$ for all $t'\in \{t_G+1,\ldots,t_F-1\}$. To show this, observe that the time $t_G$ shown to exist by the above reasoning does not satisfy this additional constraint if and only if some instance of $\vec C$ locally outputs $y(v,t')=1$ at node $v$ in such a round $t'$. The only possible such round $t'$ is $r_{k-1}$, as $t' \ge t_G + 1 > t_{k-1} > r_{k-2}$. However, in this case, each $w\in V\setminus F$ sets $x(w,r_{k-1})=0$ in round $r_{k-1}$ regardless of $m(w,r_{k-1})$ in Step 4, and we can conclude that some $w\in V\setminus F$ must set $x(w,t'')=1$ in some round $t'' \in \{r_{k-1} +1,\ldots, t_k \}$. As above, it follows that there is a round $t_G\in \{t_{k-1},\ldots,t_k-1 \}$ and a node $u\in V\setminus F$ such that $\sgo(u,t_G)=1$. Overall, we see that the safety condition for a firing squad algorithm with response time 
\[
t_F - t_G \le r_k - t_{k-1} = t_k + T(\vec C) - t_{k-1} = \Psi + T(\vec C)  = R
\]
is satisfied in rounds $t_F \ge r_1 = t_1+T(\vec C)$.

(FS3) It remains to argue that the algorithm satisfies the liveness property with response time bounded by $R$. Suppose at least $f+1$ nodes $v\in V\setminus F$ satisfy $\sgo(v,t_G)=1$ in some round $t_G\geq t_0-1$. Then, in round $t_G+1\geq t_0$ all nodes $v\in V\setminus F$ set $x(v,t_G+1)=1$ and $m(v, t_G+1)=1$ according to Step 2. Assume for contradiction that $\sfire(v,t)=0$ for all $t\in \{t_G+1,\ldots,t_G+R\}$. Denote by $t_G+1 \le t_k \leq t_G+\Psi$ the unique round such that $t_k=t_0+k\Psi$ for some $k\in \bN_0$. The instance of $\vec C$ started in this round will satisfy that all correct nodes $v\in V\setminus F$ have input $x(v,t_k)=1$: by our assumption towards contradiction, no node can locally output $y(v,t')=1$ during rounds $t' \in \{ t_G+1, \ldots, t_G + R \}$; thus, no node can set $x(v,\cdot)$ to $0$ without setting $m(v,\cdot)$ to $0$ first (by Step 3 and Step 5), which in turn entails that at time $t_k$ an instance of $\vec C$ with value of $x(v,t_k)=1$ is started before this happens. By the properties of $\vec C$, it follows that each $v\in V\setminus F$ locally outputs $1$ in round $r_k = t_k+T(\vec C) \leq t_G+\Psi+T(\vec C) \le t_G + R$, contradicting our previous assumption. We conclude that our algorithm satisfies the liveness property with response time $R=\Psi+T(\vec C)$ for rounds $t_G\geq t_0-1$. 

As $t_0\leq T(\vec P)$, it follows that the algorithm satisfies (FS1) agreement after round $t_0$, (FS2) safety after round $t_1$, and (FS3) liveness after round $t_0-1$. Since $t_1 = t_0 + \Psi \le T(\vec P) + \Psi$, it follows that the algorithm is a firing squad with response time at most $R = \Psi+T(\vec C)$ that stabilises in $\max\{T(\vec P), T(\vec P) + \Psi , T(\vec P)-1 \}= T(\vec P) +\Psi$ rounds. The bound on the message size follows from the fact that the algorithm $\vec F$ only broadcasts 1 bit in Step 1 in addition to the messages related to $\vec P$ and $\vec C$.
\end{proof}

\section{From weak pulsers to strong pulsers}\label{sec:strong-pulsers}

In \sectionref{sec:firing-squad}, we established that it suffices to construct suitable strong pulsers to solve the synchronous counting and firing squad problems. We will now reduce the construction of strong pulsers to constructing \emph{weak pulsers}.

\subsection{Weak pulsers}

A weak $\Phi$-pulser is similar to a strong pulser, but does not guarantee a fixed frequency of pulses. However, it guarantees to \emph{eventually} generate a pulse followed by $\Phi-1$ rounds of silence. Formally, we define weak pulsers as follows.

\begin{definition}[Weak pulsers]\label{def:weak}
An algorithm $\vec W$ is an $f$-resilient weak $\Phi$-pulser that stabilises in $T(\vec W)$ rounds if the following holds. In each round $t\in \bN$, each node $v\in V$ produces an output $a(v,t)$. Moreover, there exists a round $t_0\leq T(\vec W)$ such that
\begin{enumerate}[label=W\arabic*.,noitemsep]
 \item for all $v,w\in V\setminus F$ and all rounds $t\geq t_0$, $a(v,t)=a(w,t)$, 
 \item $a(v,t_0)=1$ for all $v \in V \setminus F$, and
 \item $a(v,t)=0$ for all $v \in V \setminus F$ and $t\in \{t_0+1,\ldots,t_0+\Phi-1\}$.
\end{enumerate}
We say that on round $t_0$ a \emph{good} pulse is generated by $\vec W$.
\end{definition}
\figureref{fig:weak-pulser} illustrates a weak $4$-pulser. Note that while the definition formally only asks for one good pulse, the fact that the algorithm guarantees this property for any starting state implies that there is a good pulse at least every $T(\vec W)$ rounds. 

\begin{figure}
 \begin{center}
 \includegraphics[page=4,scale=1.0]{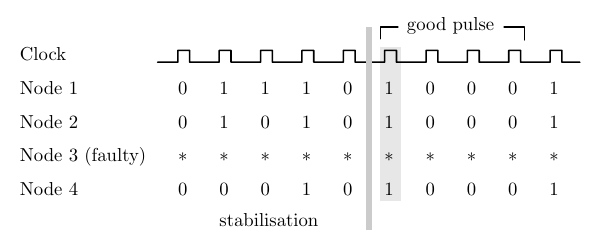}
\end{center}
\caption{An example execution of a weak 4-pulser on $n=4$ nodes with $f=1$ faulty node. Eventually, a good pulse is generated, which is highlighted. A good pulse is followed by three rounds in which no correct node generates a pulse. In contrast, the pulse two rounds earlier is not good, as it is followed by only one round of silence. \label{fig:weak-pulser}}
\end{figure}

\subsection{Constructing strong pulsers from weak pulsers}

Recall that a strong pulser can be obtained by having nodes locally count down the rounds until the next pulse, provided we have a way of ensuring that the local counters eventually agree. This can be achieved by using a weak pulser to control a suitable consensus routine, where again we always have only a single instance running at any time. While some instances will be aborted before they can complete, this will not affect the counters, as we only adjust them when the consensus routine completes. On the other hand, the weak pulser guarantees that within $T(\vec W)$ rounds, there will be a pulse followed by $\Phi-1$ rounds of silence, enabling to complete a run of any consensus routine $\vec C$ satisfying $T(\vec C)\leq \Phi$. Thus, for constructing a strong $\Psi$-pulser, we assume that we have the following $f$-resilient algorithms available:
\begin{itemize}[noitemsep]
 \item a $T(\vec C)$-round $\Psi$-value consensus algorithm $\vec C$ and
 \item a weak $\Phi$-pulser $\vec W$ for $\Phi \ge T(\vec C)$.
\end{itemize}

Given the above two algorithms, we show how to construct an $f$-resilient strong $\Psi$-pulser for any $\Psi > 1$. The pulser will stabilise in time $T(\vec W) + T(\vec C) + \Psi$ and the message size of the strong pulser will be bounded by $M(\vec W) + M(\vec C)$.

As mentioned earlier, the idea is to have nodes simply count locally between pulses and use the weak pulser to execute a single instance of the consensus algorithm $\vec C$. Eventually, a good pulse will run an instance consistently and establish agreement among the local counters. Leveraging validity, we can ensure that the counters will never be affected by the consensus routine running in the background again.

\paragraph{Variables.}
Beside the variables of the weak pulser $\vec W$ and (a single copy of) $\vec C$, our construction of a strong $\Psi$-pulser uses the following local variables:
\begin{itemize}[noitemsep]
 \item $a(v,t) \in \{0,1\}$ is the output variable of the weak $\Phi$-pulser $\vec W$,
 \item $b(v,t) \in \{0,1\}$ is the output variable of the strong $\Psi$-pulser we are constructing,
 \item $c(v,t) \in [\Psi]$ is the local counter keeping track on when the next pulse occurs, and
 \item $d(v,t) \in \{1,\ldots,T(\vec C)\} \cup \{ \bot \} $ keeps track of how many rounds an instance of $\vec C$ has been executed since the last pulse from the weak pulser $\vec W$. The value $\bot$ denotes that the consensus routine has stopped.
\end{itemize}
\paragraph{Strong pulser algorithm.}
The algorithm is as follows. Each node $v$ executes the weak $\Phi$-pulser algorithm~$\vec W$ in addition to the following instructions on each round $t \in \bN$:
\begin{enumerate}[noitemsep]
 \item If $c(v,t) = 0$, then set $b(v,t)=1$ and otherwise $b(v,t)=0$.
 \item Set $c'(v,t) = c(v,t)$.
 \item If $d(v,t) \neq \bot$, then 
 \begin{enumerate}
  \item Execute the instructions of $\vec C$ for round $d(v,t)$.
  \item If $d(v,t)\neq T(\vec C)$, set $d(v,t+1) = d(v,t) + 1$.
  \item If $d(v,t)=T(\vec C)$, then
  \begin{enumerate}
    \item Set $c'(v,t) = y(v,t) + T(\vec C) \bmod \Psi$, where $y(v,t)$ is the output value of $\vec C$.
    \item Set $d(v,t+1)=\bot$.
  \end{enumerate}
 \end{enumerate}
 \item Update $c(v,t+1) = c'(v,t) + 1 \bmod \Psi$.
 \item If $a(v,t)=1$, then 
 \begin{enumerate}
  \item Start a new instance of $\vec C$ using $c'(v,t)$ as input (resetting all state variables of $\vec C$).
  \item Set $d(v,t+1)=1$.
 \end{enumerate}
\end{enumerate}

In the above algorithm, the first step simply translates the counter value to the output of the strong pulser. We then use a temporary variable $c'(v,t)$ to hold the counter value, which is overwritten by the output of $\vec C$ (increased by $T(\vec C)\bmod \Psi$) if it completes a run in this round. In either case, the counter value needs to be increased by $1\bmod \Psi$ for the next round. The remaining code does the bookkeeping for an ongoing run of $\vec C$ and starting a new run if the weak pulser generates a pulse.

Observe that in the above algorithm, each node only sends messages related to the weak pulser $\vec W$ and the consensus algorithm $\vec C$. Thus, there is no additional overhead in communication and the message size is bounded by $M(\vec W) + M(\vec C)$. Hence, it remains to show that the local counters $c(v,t)$ implement a strong $\Psi$-counter.

\begin{theorem}\label{thm:strong-agreement}
The variables $c(v,t)$ in the above algorithm implement a synchronous $\Psi$-counter that stabilises in $T(\vec W) + T(\vec C) + 1$ rounds and uses messages of at most $M(\vec W) + M(\vec C)$ bits.
\end{theorem}
\begin{proof}
Suppose round $t_0\leq T(\vec W)$ is as in \defref{def:weak}, that is, $a(v,t)=a(w,t)$ for all $t\geq t_0$, and a good pulse is generated in round $t_0$. Thus, all correct nodes participate in simulating an instance of $\vec C$ during rounds $t_0+1,\ldots,t_0+T(\vec C)$, since no pulse is generated during rounds $t_0+1,\ldots,t_0+T(\vec C)-1$, and thus, also no new instance is started in the last step of the code during these rounds. 

By the agreement property of the consensus routine, it follows that $c'(v,t_0+T(\vec C))=c'(w,t_0+T(\vec C))$ for all $v,w\in V\setminus F$ after Step 3ci. By Steps 2 and 4, the same will hold for both $c(\cdot,t')$ and $c'(\cdot,t')$, $t'>t_0+T(\vec C)$, provided that we can show that in rounds $t'>t$, Step 3ci never sets $c'(v,t)$ to a value different than $c(v,t)$ for any $v\in V\setminus F$; as this also implies that $c(v,t'+1)=c(v,t')+1\bmod \Psi$ for all $v\in V\setminus F$ and $t'>t_0+T(\vec C)$, this will complete the proof.

Accordingly, consider any execution of Step 3ci in a round $t'>t_0+T(\vec C)$. The instance of $\vec C$ terminating in this round was started in round $t'-T(\vec C)>t_0$. However, in this round the weak pulser must have generated a pulse, yielding that, in fact, $t'-T(\vec C)\geq t_0+T(\vec C)$. Assuming for contradiction that $t'$ is the earliest round in which the claim is violated, we thus have that $c'(v,t'-T(\vec C))=c'(w,t'-T(\vec C))$ for all $v,w\in V\setminus F$, i.e., all correct nodes used the same input value $c$ for the instance. By the validity property of $\vec C$, this implies that $v\in V\setminus F$ outputs $y(v,t')=c$ in round $t'$ and sets $c'(v,t')=c+T(\vec C)\bmod \Psi$. However, since $t'$ is the earliest round of violation, we already have that $c'(v,t')=c(v,t')=c+T(\vec C)\bmod \Psi$ after the second step, contradicting the assumption and showing that the execution stabilised in round $t_0+T(\vec C)+1\leq T(\vec W)+T(\vec C)+1$.
\end{proof}

Together with \lemmaref{lemma:strong_pulsers_equal_counters}, we get the following corollary.

\begin{corollary}\label{cor:strong-pulsers}
Let $\Psi > 1$. Suppose there exists an $f$-resilient $\Psi$-value consensus routine $\vec C$ and a weak $\Phi$-pulser $\vec W$, where $\Phi \ge T(\vec C)$. Then there exists an $f$-resilient strong $\Psi$-pulser $\vec P$ that
\begin{itemize}[noitemsep]
 \item stabilises in time $T(\vec P) \leq T(\vec C) + T(\vec W) + \Psi$, and
 \item uses message of size at most $M(\vec P) \leq M(\vec C) + M(\vec W)$ bits.
\end{itemize}
\end{corollary}

\section{Constructing weak pulsers from less resilient strong pulsers}\label{sec:weak-pulsers}

Having seen that we can construct strong pulsers from weak pulsers using a consensus algorithm, the only piece missing in our framework is the existence of efficient weak pulsers. Indeed, having a pair of an $f$-resilient weak pulser and a consensus routine, we immediately obtain a corresponding firing squad algorithm.

In this section, we devise a recursive construction of a weak pulser from strong pulsers of smaller resilience. Given that a $0$-resilient pulser is trivial and that we can obtain strong pulsers from weak ones without losing resilience, this is sufficient for constructing strong pulsers of optimal resilience from consensus algorithms of optimal resilience.

Our approach bears similarity to our constructions from earlier work \cite{lenzen15efficient,lenzen15towards}, but attains better bit complexity and can be used with an arbitrary consensus routine. On a high level, we take the following approach as also illustrated in \figureref{fig:filter-overview}:
\begin{enumerate}
 \item Partition the network into two parts, each running a strong pulser (with small resilience). Our construction guarantees that at least one of the strong pulsers stabilises.
 \item Filtering of pulses generated by the strong pulsers:
 \begin{enumerate}[label=(\alph*)]
 \item Nodes consider the observed pulses generated by the strong pulsers as \emph{potential} pulses.
 \item Since one of the strong pulsers may not stabilise, it may generate \emph{spurious} pulses, that is, pulses that only a subset of the correct nodes observe. 
 \item We limit the \emph{frequency} of the spurious pulses using a filtering mechanism based on threshold voting.
 \end{enumerate}
 \item We enforce any spurious pulse to be observed by all correct nodes by employing a \emph{silent consensus} routine. In silent consensus, no message is sent (by correct nodes) if all correct nodes have input $0$. Thus, if all nodes actually participating in an instance have input $0$, non-participating nodes behave \emph{as if they participated} with input $0$. This avoids the chicken-and-egg problem of having to solve consensus on participation in the consensus routine. We make sure that if any node uses input $1$, i.e., the consensus routine may output 1, all nodes participate. Thus, when a pulse is generated, all correct nodes agree on this.
 \item If a potential pulse generated by one of the pulsers both passes the filtering step and the consensus instance outputs ``1'', then a weak pulse is generated.
\end{enumerate}

\begin{figure}
\begin{center}
 \includegraphics[page=6,scale=1.0]{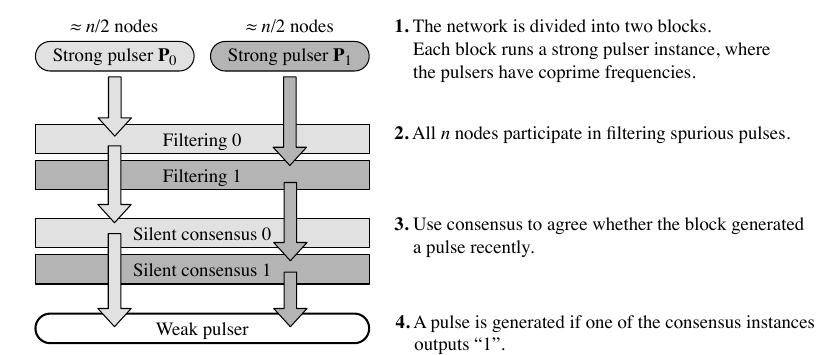}
\end{center}
\caption{Overview of the weak pulser construction. Light and dark grey boxes correspond to steps of block 0 and 1, respectively. The small rounded boxes denote the pulser algorithms $\vec P_i$ that are run (in parallel) on two disjoint sets of roughly $n/2$ nodes, whereas the wide rectangular boxes denote to the filtering steps in which all of the $n$ nodes are employed. The arrows indicate the flow of information for each block.\label{fig:filter-overview}}
\end{figure}

\subsection{The filtering construction}

Our goal is to construct a weak $\Phi$-pulser (for sufficiently large $\Phi$) with resilience $f$. We partition the set of $n$ nodes into two disjoint sets $V_0$ and $V_1$ with $n_0$ and $n_1$ nodes, respectively. Thus, we have $n = n_0 + n_1$. For $i \in \{0,1\}$, let $\vec P_i$ be an $f_i$-resilient strong $\Psi_i$-pulser. That is, $\vec P_i$ generates a pulse every $\Psi_i$ rounds once stabilised, granted that $V_i$ contains at most $f_i$ faulty nodes. Nodes in block $i$ execute the algorithm $\vec P_i$. Our construction tolerates $f = f_0 + f_1 + 1$ faulty nodes. Since we consider Byzantine faults, we require the additional constraint that $f < n/3$. 

Let $a_i(v,t) \in \{0,1\}$ indicate the output bit of $\vec P_i$ for a node $v \in V_i$. Note that we might have a block $i \in \{0,1\}$ that contains more than $f_i$ faulty nodes. Thus, it is possible that the algorithm $\vec P_i$ never stabilises. In particular, we might have the situation that some of the nodes in block $i$ produce a pulse, but others do not. We say that a pulse generated by such a $\vec P_i$ is \emph{spurious.} We proceed by showing how to filter out such spurious pulses if they occur too often.

\paragraph{Filtering rules.}
We define five variables with the following semantics:
\begin{itemize}[noitemsep]
 \item $m_i(v,t+1)$ indicates whether at least $n_i-f_i$ nodes $u \in V_i$ sent $a_i(u,t)=1$,
 \item $M_i(v,t+1)$ indicates whether at least $n-f$ nodes $u \in V$ sent $m_i(u,t)=1$,
 \item $\ell_i(v,t)$ indicates when was the last time block $i$ triggered a (possibly spurious) pulse, 
 \item $w_i(v,t)$ indicates how long any firing events coming from block $i$ are ignored, and
 \item $b_i(v,t)$ indicates whether node $v$ accepts a firing event from block $i$.
\end{itemize}
The first two of the above variables are set according to the following rules:
\begin{itemize}[noitemsep]
 \item $m_i(v,t+1) = 1$ if and only if $|\{ u \in V_i : a_i(v, u,t) = 1 | \} \ge n_i - f_i$, 
 \item $M_i(v,t+1) = 1$ if and only if $|\{ u \in V : m_i(v, u,t) = 1 \} \ge n-f$,
\end{itemize}
where $a_i(v,u,t)$ and $m_i(v,u,t)$ denote the values for $a(\cdot)$ and $m(\cdot)$ node $v$ received from $u$ at the end of round $t$, respectively. Furthermore, we update the $\ell(\cdot,\cdot)$ variables using the rule
\[
 \ell_i(v,t+1) = \begin{cases}
                0 & \text{if } |\{ u\in V : m_i(u,t)=1\}| \ge f+1, \\
                \min \{ \Psi_i, \ell_i(v,t) + 1  \} & \text{otherwise.} 
              \end{cases}
\]
In words, the counter is reset on round $t+1$ if $v$ has proof that at least one correct node $u$ had $m_i(u,t)=1$, that is, some $u$ observed $\vec P_i$ generating a (possibly spurious) pulse.

We reset the cooldown counter $w_i$ whenever suspicious activity occurs. The idea is that it is reset to its maximum value $C$ by node $v$ in the following two cases:
\begin{itemize}[noitemsep]
 \item some other correct node $u \neq v$ observed block $i$ generating a pulse, but the node $v$ did not 
 \item block $i$ generated a pulse, but this happened either too soon or too late.
\end{itemize}
To capture this behaviour, the cooldown counter is set with the rule
\[
 w_i(v,t+1) = \begin{cases}
                C & \text{if } M_i(v,t+1) = 0 \text{ and } \ell_i(v,t+1) = 0, \\
                C & \text{if } M_i(v,t+1) = 1 \text{ and } \ell_i(v,t) \neq \Psi_i - 1, \\
                \max \{ w_i(v,t) - 1 , 0 \} & \text{otherwise,}
              \end{cases}
\]
where $C = \max\{ \Psi_0, \Psi_1 \} + \Phi + 2$. Finally, a node $v$ accepts a pulse generated by block $i$ if the node's cooldown counter is zero and it saw at least $n-f$ nodes supporting the pulse. The variable $b_i(v,t)$ indicates whether node $v$ accepted a pulse from block $i$ on round $t$. The variable is set using the rule
\[
 b_i(v,t) = \begin{cases}
                1 & \text{if } w_i(v,t) = 0 \text{ and } M_i(v,t)=1, \\
                0 & \text{otherwise.}
              \end{cases}
\]

\subsection{Analysis of the filtering construction}

We now analyse when the nodes accept firing events generated by the blocks. We say that a block $i$ is correct if it contains at most $f_i$ faulty nodes. Note that since there are at most $f = f_0 + f_1 + 1$ faulty nodes, at least one block $i \in \{0,1\}$ will be correct. Thus, eventually the algorithm $\vec P_i$ run by a correct block $i$ will stabilise. This yields the following lemma.

\begin{lemma}\label{lemma:one_correct}
For some $i \in \{0,1\}$, the strong pulser algorithm $\vec P_i$ stabilises by round $T(\vec P_i)$.
\end{lemma}

We proceed by establishing some bounds on when (possibly spurious) pulses generated by block $i$ are accepted. We start with the case of having a correct block $i$.

\begin{lemma}
If block $i$ is correct, then there exists a round $r_0 \le T(\vec P_i) + 2$ such that for each $v\in V\setminus F$, $M_i(v,t) = 1$ if and only if $t = r_0 + k\Psi_i$ for $k \in \bN_0$. \label{lemma:correct-block-votes}
\end{lemma}
\begin{proof}
If block $i$ is correct, then the algorithm $\vec P_i$ stabilises by round $T(\vec P_i)$. Hence, there is some $t_0 \leq T(\vec P)$ so that the output variable $a_i(\cdot)$ of $\vec P_i$ satisfies
\[
 a_i(v, t) = 1 \text{ if and only if } t = t_0 + k \Psi_i \text{ for } k \in \bN_0
\]
holds for all $t \geq t_0$. We will now argue that $r_0 = t_0 + 2$ satisfies the claim of the lemma.

If $\vec P_i$ generates a pulse on round $t\geq t_0$, then at least $n_i - f_i$ correct nodes $u \in V_i \setminus F$ have $a_i(u,t)=1$. Therefore, for all $v \in V \setminus F$ we have $m_i(v,t+1) = 1$, and consequently, $M_i(v,t+2) = 1$. Since block $i$ is correct, there are at most $f_i$ faulty nodes in the set $V_i$. Observe that by \lemmaref{lemma:strong_pulsers_equal_counters} strong pulsers solve synchronous counting, which in turn is as hard as consensus \cite{dolev15survey}. This implies that we must have $f_i < n_i/3$, as $\vec P_i$ is a strong $f_i$-resilient pulser for $n_i$ nodes. Therefore, if $\vec P_i$ does not generate a pulse on round $t\geq t_0$, then at most $f_i < n_i - f_i$ faulty nodes $u$ may claim $a_i(u,t)=1$. This yields that $m_i(v,t+1) = M_i(v,t+2) = 0$ for all $v \in V \setminus F$.
\end{proof}

We can now establish that a correct node accepts a pulse generated by a correct block $i$ exactly every $\Psi_i$ rounds. 

\begin{lemma}
\label{lemma:correct-block-fires}
If block $i$ is correct, then there exists a round $t_0 \le T(\vec P_i) + 2C$ such that for each $v\in V\setminus F$, $b_i(v,t) = 1$ for any $t \ge t_0$ if and only if $t = t_0 + k\Psi_i$ for $k \in \bN_0$.
\end{lemma}
\begin{proof}
 \lemmaref{lemma:correct-block-votes} implies that there exists $r_0 \le T(\vec P_i) + 2$ such that both $M_i(v,t) = 1$ and $\ell_i(v,t) = 0$ hold for $t \ge r_0$ if and only if $t = r_0 + k\Psi_i$ for $k \in \bN_0$. Thus, it follows that $w_i(v,t+1) = \max\{w_i(v,t)-1,0\}$ for all such $t$ and hence $w_i(v,t')=0$ for all $t' \ge r_0 + C + 2$. The claim now follows from the definition of $b_i(v,t')$, the choice of $r_0$, and the fact that $\Psi_i \leq C-2$.
\end{proof}

It remains to deal with the faulty block. If we have Byzantine nodes, then a block $i$ with more than $f_i$ faulty nodes may attempt to generate spurious pulses. However, the filtering mechanism prevents the spurious pulses from occuring too frequently.

\begin{lemma}\label{lemma:any-block-fires}
Let $v, v' \in V \setminus F$ and $t>2$. Suppose $b_i(v,t)=1$ and suppose that $t'>t$ is minimal such that $b_i(v',t')=1$. Then $t' = t + \Psi_i$ or $t' > t+C$.
\end{lemma}
\begin{proof}
Suppose $b_i(v,t)=1$ for some correct node $v \in V$ and $t>2$. Since $b_i(v,t)=1$, $w_i(v,t)=0$ and $M_i(v,t)=1$. Because $M_i(v,t)=1$, there must be at least $n-2f>f$ correct nodes $u$ such that $m_i(u,t-1)=1$. Hence, $\ell_i(u,t)=0$ for every node $u \in V \setminus F$. 

Recall that $t'>t$ is minimal so that $b_i(v',t')=1$. Again, $w_i(v',t')=0$ and $M_i(v',t')=1$. Moreover, since $\ell_i(v',t)=0$, we must have $\ell_i(v',r) < \Psi_i - 1$ for all $t \le r < t + \Psi_i - 1$. This implies that $t' \ge t + \Psi_i$, as $w_i(v',t')=0$ and $M_i(v',t')=1$ necessitate that $\ell_i(v',t'-1)=\Psi_i-1$. In the event that $t' \neq t + \Psi_i$, the cooldown counter must have been reset at least once, i.e., $w_i(v',r)=C$ holds for some $t < r \le t' - C$, implying that $t' > t + C$.
\end{proof}

\subsection{Introducing silent consensus}

The above filtering mechanism prevents spurious pulses from occurring too often: if some node accepts a pulse from block $i$, then no node accepts a pulse from this block for at least $\Psi_i$ rounds. We now strengthen the construction to enforce that any (possibly spurious) pulse generated by block $i$ will be accepted by either all or no correct nodes. 
In order to achieve this, we employ \emph{silent consensus.}
\begin{definition}[Silent consensus]
We call a consensus protocol \emph{silent}, if in each execution in which all
correct nodes have input $0$, correct nodes send no messages.
\end{definition}
The idea is that this enables to have consistent executions even if not all correct nodes actually take part in an execution, provided we can ensure that in this case all participating correct nodes use input 0: the non-participating nodes send no messages either, which is the exact same behavior participating nodes would exhibit. We show that silent consensus protocols can be obtained from non-silent ones using a simple transformation.
\begin{theorem}\label{thm:silent-consensus}
Any consensus protocol $\vec C$ can be transformed into a silent binary consensus protocol $\vec C'$ with $T(\vec C') = T(\vec C)+2$ and the same resilience and message size.
\end{theorem}
\begin{proof}
The new protocol $\vec C'$ can be seen as a ``wrapper'' protocol that manipulates the inputs and then lets each node decide whether it participates in an instance of the original protocol. The output of the original protocol, $\vec C$, will be taken into account only by correct nodes that participate throughout the protocol, as specified below.

In the first round of the new protocol, $\vec C'$, each participating node broadcasts its input if it is $1$ and otherwise sends nothing. If a node receives fewer than $n-f$ times the value $1$, it sets its input to $0$. In the second round, the same pattern is applied.

Subsequently, $\vec C$ is executed by all nodes that received at least $f+1$ messages in the first round. If during the execution of $\vec C$ a node
\begin{enumerate}[label=(\roman*)]
 \item cannot process the messages received in a given round in accordance with $\vec C$ (this may happen e.g. when not all of the correct nodes participate in the instance, which is not covered by the model assumptions of $\vec C$), 
 \item would have to send more bits than it would have according to the known bound $M(\vec C)$, or
 \item would violate the running time bound of $\vec C$,
\end{enumerate}
then the node (locally) aborts the execution of $\vec C$. 
 Finally, a node outputs $0$ in the new protocol if it did not participate in the execution of $\vec C$, aborted it, or received $f$ or fewer messages in the second round, and it outputs the result according to the run of $\vec C$ otherwise.

We first show that the new protocol, $\vec C'$, is a consensus protocol with the same resilience as $\vec C$ and the claimed bounds on communication complexity and running time. We distinguish two cases. First, suppose that all correct nodes participate in the execution of $\vec C$ at the beginning of the third round. As all nodes participate, the bounds on resilience, communication complexity, and running time that apply to $\vec C$ hold in this execution, and no node will quit executing the protocol before termination. To establish agreement and validity, again we distinguish two cases. If all nodes output the outcome of the execution of $\vec C$, these properties follow right away since $\vec C$ satisfies them; here we use that although the initial two rounds might affect the inputs of nodes, a node will change its input to $0$ only if there is at least one correct node with input $0$. On the other hand, if some node outputs $0$ because it received $f$ or fewer messages in the second round of $\vec C'$, no node received more than $2f<n-f$ messages in the second round. Consequently, all nodes executed $\vec C$ with input $0$ and computed output $0$ by the agreement property of $\vec C$, implying agreement and validity of
the new protocol.

The second case is that some correct node does not participate in the execution of $\vec C$. Thus, it received at most $f$ messages in the first round of $\vec C'$, implying that no node received more than $2f<n-f$ messages in this round. Consequently, correct nodes set their input to $0$ and will not transmit in the second round. While some nodes may execute $\vec C$, all correct nodes will output $0$ no matter how $\vec C$ behaves. Since nodes abort the execution of $\vec C$ if the bounds on communication or time complexity are about to be violated, the claimed bounds for the new protocol hold.

It remains to show that the new protocol is silent. Clearly, if all correct nodes have input $0$, they will not transmit in the first two rounds. In particular, they will not receive more than $f$ messages in the first round and not participate in the execution of $\vec C$. Hence correct nodes do not send messages at all, as claimed.
\end{proof}

For example, plugging in the phase king protocol~\cite{berman89consensus}, we get the following corollary.

\begin{corollary}
For any $f < n/3$, there exists an $f$-resilient deterministic silent binary consensus protocol $\vec C$ with $T(\vec C) \in \Theta(f)$ and $M(\vec C) \in O(1)$.
\end{corollary}

\subsection{Using silent consensus to prune spurious pulses}

As the filtering construction bounds the frequency at which spurious pulses may occur from above, we can make sure that at each time, only one consensus instance can be executed for each block. However, we need to further preprocess the inputs, in order to make sure that (i) all correct nodes participate in an instance or (ii) no participating correct node has input 1; here, output 1 means agreement on a pulse being triggered, while output 0 results in no action.

Recall that $b_i(v,t)\in \{0,1\}$ indicates whether $v$ observed a (filtered) pulse of the strong pulser $\vec P_i$ in round $t$. Moreover, assume that $\vec C$ is a silent consensus protocol running in $T(\vec C)$ rounds. We use two copies $\vec C_i$, where $i\in \{0,1\}$, of the consensus routine $\vec C$. We require that $\Psi_i \geq T(\vec C)$, which guarantees by \lemmaref{lemma:any-block-fires} that (after stabilisation) every instance of $\vec C$ has sufficient time to complete. Adding one more level of voting to clean up the inputs, we arrive at the following routine.

\paragraph{The pruning algorithm.}
Besides the local variables of $\vec C_i$, the algorithm will use the following variables for each $v \in V$ and round $t \in \bN$:
\begin{itemize}[noitemsep]
 \item $y_i(v,t)\in \{0,1\}$ denotes the output value of consensus routine $\vec C_i$, 
 \item $r_i(v,t)\in \{1,\ldots,T(\vec C)\}\cup \{\bot\}$ is a local round counter for controlling $\vec C_i$, and 
 \item $B_i(v,t) \in \{0,1\}$ is the output of block $i$. 
\end{itemize}
Now each node $v$ executes the following on round $t$:
\begin{enumerate}[noitemsep]
 \item Broadcast the value $b_i(v,t)$.
 \item If $b_i(v,w,t-1)=1$ for at least $n-2f$ nodes $w\in V$, then reset $r_i(v,t)=1$.
 \item If $r_i(v,t)=1$, then 
 \begin{enumerate}[noitemsep]
  \item start a new instance of $\vec C_i$, that is, re-initialise the variables of $\vec C_i$ correctly,
  \item use input $1$ if $b_i(v,w,t-1)=1$ for at least $n-f$ nodes $w\in V$ and $0$ otherwise.
 \end{enumerate}
 \item If $r_i(v,t)=T(\vec C)$, then 
 \begin{enumerate}[noitemsep]
  \item execute round $T(\vec C)$ of $\vec C_i$,
  \item set $r_i(v,t+1) = \bot$, 
  \item set $B_i(v,t+1)=y_i(v,t)$, where $y_i(v,t) \in \{0,1\}$ is the output variable of $\vec C_i$.
 \end{enumerate}
 Otherwise, set $B_i(v,t+1)=0$.
 \item If $r_i(v,t)\not \in \{T(\vec C),\bot\}$, then 
 \begin{enumerate}[noitemsep]
  \item execute round $r_i(v,t)$ of $\vec C_i$, and
  \item set $r_i(v,t+1)=r_i(v,t)+1$.
 \end{enumerate}
\end{enumerate}

\paragraph{Analysis.}
Besides the communication used for computing the values $b_i(\cdot)$, the above algorithm uses messages of size $M(\vec C)+1$, as $M(\vec C)$ bits are used when executing $\vec C_i$ and one bit is used to communicate the value of $b_i(v,t)$.

We say that $v\in V\setminus F$ \emph{executes round $r\in \{1,\ldots,T(\vec C)$\}} of $\vec C_i$ in round $t$ iff $r_i(v,t)=r$. By \lemmaref{lemma:any-block-fires}, in rounds $t>T(\vec C)+2$, there is always at most one instance of $\vec C_i$ being executed, and if so, consistently.
\begin{corollary}\label{coro:execute_consistently}
Suppose $v\in V\setminus F$ executes round $1$ of $\vec C_i$ in some round $t>T(\vec C)+2$. Then there is a subset $U\subseteq V\setminus F$ such that each $u\in U$ executes round $r\in \{1,\ldots,T(\vec C)\}$ of $\vec C_i$ in round $t+r-1$ and no $u\in V\setminus (F\cup U)$ executes any round of $\vec C_i$ in round $t+r-1$.
\end{corollary}
Exploiting silence of $\vec C_i$ and the choice of inputs, we can ensure that the case $U\neq V\setminus F$ causes no trouble.
\begin{lemma}\label{lemma:full_participation}
Let $t>T(\vec C)+2$ and $U$ be as in \corollaryref{coro:execute_consistently}. Then $U=V\setminus F$ or each $u\in U$ has input $0$ for the respective instance of $\vec C_i$. 
\end{lemma}
\begin{proof}
Suppose $u\in U$ starts an instance with input $1$ in round $t'\in \{t-T(\vec C)-1,\ldots,t\}$. Then $b_i(w,t'-1)=1$ for at least $n-2f$ nodes $w\in V\setminus F$, since $u$ received $b_i(u,w,t'-1)=1$ from $n-f$ nodes $w\in V$. Thus, each $v\in V\setminus F$ received $b_i(v,w,t'-1)=1$ from at least $n-2f$ nodes $w$ and sets $r_i(v,t')=1$, i.e., $U=V\setminus F$. The lemma now follows from \corollaryref{coro:execute_consistently}.
\end{proof}
Recall that if all nodes executing $\vec C_i$ have input $0$, non-participating correct nodes behave exactly as if they executed $\vec C_i$ as well, i.e., they send no messages. Hence, if $U\neq V\setminus F$, all nodes executing the algorithm will compute output $0$. Therefore, \corollaryref{coro:execute_consistently}, \lemmaref{lemma:any-block-fires}, and \lemmaref{lemma:full_participation} imply the following corollary.
\begin{corollary}\label{coro:full_participation}
In rounds $t>T(\vec C)+2$ it holds that $B_i(v,t)=B_i(w,t)$ for all $v,w\in V\setminus F$ and $i\in \{0,1\}$. Furthermore, if $B_i(v,t)=1$ for $v\in V\setminus F$ and $t>T(\vec C)+2$, then the minimal $t'>t$ so that $B_i(v,t')=1$ (if it exists) satisfies either $t'=t+\Psi_i$ or $t'>t+C=t+\max\{\Psi_0,\Psi_1\}+\Phi+2$.
\end{corollary}
Finally, we observe that our approach does not filter out pulses from correct blocks.
\begin{lemma}\label{lemma:correct_passes}
If block $i$ is correct, there is a round $t_0\leq T(\vec P_i)+2C+T(\vec C)+1$ so that for any $t\geq t_0$, $B_i(v,t)=1$ if and only if $t=t_0+k\Psi_i$ for some $k\in \bN_0$.
\end{lemma}
\begin{proof}
\lemmaref{lemma:correct-block-fires} states the same for the variables $b_i(v,t)$ and a round $t_0'\leq T(\vec P_i)+2C$. If $b_i(v,t)=1$ for all $v\in V\setminus F$ and some round $t$, all correct nodes start executing an instance of $\vec C_i$ with input $1$ in round $t+1$. As, by \corollaryref{coro:execute_consistently}, this instance executes correctly and, by validity of $\vec C_i$, outputs $1$ in round $t+T(\vec C)$, all correct nodes satisfy $B_i(v,t+T(\vec C)+1)=1$. Similarly, $B_i(v,t+T(\vec C)+1)=0$ for such $v$ and any $t\geq t_0'$ with $b_i(v,t)=0$.
\end{proof}

\subsection{Obtaining the weak pulser}

Finally, we define the output variable of our weak pulser as
\[
 B(v,t) = \max\{B_0(v,t),B_1(v,t)\}.
\]
As we have eliminated the possibility that $B_i(v,t)\neq B_i(w,t)$ for $v,w\in V\setminus F$ and $t>T(\vec C)+2$, Property W1 holds. Since there is at least one correct block $i$ by \lemmaref{lemma:one_correct}, \lemmaref{lemma:correct_passes} shows that there will be good pulses (satisfying Properties W2 and W3) regularly, unless block $1-i$ interferes by generating pulses violating Property W3 (i.e., in too short order after a pulse generated by block $i$). Here the filtering mechanism comes to the rescue: as we made sure that pulses are either generated at the chosen frequency $\Psi_i$ or a long period of $C$ rounds of generating no pulse is enforced (\corollaryref{coro:full_participation}), it is sufficient to choose $\Psi_0$ and $\Psi_1$ as coprime multiples of $\Phi$.

Accordingly, we pick $\Psi_0=2\Phi$ and $\Psi_1=3\Phi$ and observe that this results in a good pulse within $O(\Phi)$ rounds after the $B_i$ stabilised.

\begin{lemma}\label{lemma:good_pulse}
In the construction described in the previous two subsections, choose $\Psi_0=2\Phi$ and $\Psi_1=3\Phi$ for any $\Phi \geq T(\vec C)$. Then $B(v,t)$ is the output variable of a weak $\Phi$-pulser with stabilisation time $\max\{T(\vec P_0),T(\vec P_1)\}+O(\Phi)$.
\end{lemma}
\begin{proof}
We have that $C=\max\{\Psi_0,\Psi_1\}+\Phi+2\in O(\Phi)$. By the above observations, there is a round $t\in \max\{T(\vec P_0),T(\vec P_1)\}+T(\vec C)+O(\Phi)=\max\{T(\vec P_0),T(\vec P_1)\}+O(\Phi)$ satisfying the following four properties. For either block $i \in \{0,1\}$, we have by \corollaryref{coro:full_participation} that
\begin{enumerate}[noitemsep]
 \item $B_i(v,t') = B_i(w,t')$ and $B(v,t')=B(w,t')$ for any $v,w \in V \setminus F$ and $t' \ge t$.
\end{enumerate}
Moreover, for a correct block $i$ and for all $v \in V \setminus F$ we have from \lemmaref{lemma:correct_passes} that 
\begin{enumerate}[noitemsep,resume]
  \item $B_i(v,t)=B_i(v,t+\Psi_i)=1$,
  \item $B_i(v,t')=0$ for all $t'\in \{t+1,\ldots,t+\Phi-1\}\cup \{t+\Psi_i+1,\ldots,t+\Psi_i+\Phi-1\}$,
\end{enumerate}
and for a (possibly faulty) block $1-i$ we have from \corollaryref{coro:full_participation} that 
\begin{enumerate}[noitemsep,resume]
 \item if $B_{1-i}(v,t')=1$ for some $v \in V \setminus F$ and $t'\in \{t+1,\ldots,t+\Psi_i+\Phi-1\}$, then $B_{1-i}(u,t'')=0$ for all $u \in V \setminus F$ and $t''\in \{t'+1,\ldots,t'+C\}$ that  do not satisfy $t''=t'+k\Psi_{1-i}$ for some $k\in \N_0$.
\end{enumerate}

Now it remains to argue that a good pulse is generated. Suppose that $i$ is a correct block given by \lemmaref{lemma:one_correct}. By the first property, it suffices to show that a good pulse occurs in round $t$ or in round $t+\Psi_i$. From the second property, we get for all $v\in V\setminus F$ that $B(v,t)=1$ and $B(v,t+\Psi_i)=1$. If the pulse in round $t$ is good, the claim holds. Hence, assume that there is a round $t'\in \{t+1,\ldots,t+\Psi_i-1\}$ in which another pulse occurs, that is, $B(v,t')=1$ for some $v \in V \setminus F$. This entails that $B_{1-i}(v,t')=1$ by the third property. We claim that in this case the pulse in round $t+\Psi_i$ is good. To show this, we exploit the fourth property. Recall that $C>\Psi_i+\Phi$, i.e., $t'+C>t+\Psi_i+\Phi$. We distinguish two cases:
\begin{itemize}
  \item In the case $i=0$, we have that $t'+\Psi_{1-i}=t'+3\Phi=t'+\Psi_0+\Psi>t+\Psi_0+\Phi$, that is, the pulse in round $t+\Psi_0=t+\Psi_i$ is good.
  \item In the case $i=1$, we have that $t'+\Psi_{1-i}=t'+2\Phi< t+3\Phi= t+\Psi_1$ and $t'+2\Psi_{1-i}=t'+4\Phi=t'+\Psi_1+\Phi>t+\Psi_1+\Phi$, that is, the pulse in round $t+\Psi_1=t+\Psi_i$ is good.
\end{itemize}
In either case, a good pulse occurs by round $t+\max\{\Psi_0,\Psi_1\}\in \max\{T(\vec P_0),T(\vec P_1)\}+O(\Phi)$.
\end{proof}

From the above lemma and the constructions discussed in this section, we get the following theorem.

\begin{theorem}\label{thm:weak-pulsers}
Let $n = n_0 + n_1$ and $f = f_0 + f_1 + 1$, where $n>3f$. Suppose $\vec C$ is an $f$-resilient consensus algorithm on $n$ nodes and let $\Phi \ge T(\vec C)+2)$. If there exist $f_i$-resilient strong $\Psi_i$-pulser algorithms on $n_i$ nodes, where $\Psi_0 = 2\Phi$ 
and $\Psi_1 = 3\Phi$, then there exists an $f$-resilient weak $\Phi$-pulser $\vec W$ on $n$ nodes that satisfies
\begin{itemize}[noitemsep]
 \item $T(\vec W) \in \max\{T(\vec P_0),T(\vec P_1)\} + O(\Phi)$,
 \item $M(\vec W) \in \max\{M(\vec P_0),M(\vec P_1)\} + O(M(\vec C))$.
\end{itemize}
\end{theorem}
\begin{proof}
By \theoremref{thm:silent-consensus}, we can transform $\vec C$ into a silent consensus protocol $\vec C'$, at the cost of increasing its round complexity by $2$. Using $\vec C'$ in the construction, \lemmaref{lemma:good_pulse} shows that we obtain a weak $\Phi$-pulser with the stated stabilisation time, which by construction tolerates $f$ faults. Concerning the message size, note that we run $\vec P_0$ and $\vec P_1$ on disjoint node sets. Apart from sending $\max\{M(\vec P_0),M(\vec P_1)\}$ bits per round for its respective strong pulser, each node may send $M(\vec C)$ bits each to each other node for the two copies $\vec C_i$ of $\vec C$ it runs in parallel, plus a constant number of additional bits for the filtering construction including its outputs $b_i(\cdot,\cdot)$.
\end{proof}

\section{Main results}\label{sec:main}

Finally, in this section we put the developed machinery to use. As our main result, we show how to recursively construct strong pulsers out of consensus algorithms. 

\begin{theorem}\label{thm:consensus-to-pulsers}
Suppose that we are given a family of $f$-resilient deterministic consensus algorithms $\vec C(f)$ running on any number $n>3f$ of nodes in $T(\vec C(f))$ rounds using $M(\vec C(f))$-bit messages, where $T(\vec C(f))$ and $M(\vec C(f))$ are non-decreasing in $f$. Then, for any $\Psi \in \bN$, $f\in \bN_0$, and $n>3f$, there exists a strong $\Psi$-pulser $\vec P$ on $n$ nodes that
\begin{itemize}[noitemsep]
  \item stabilises in time $T(\vec P) \in (1+o(1))\Psi + O\left(\sum_{j=0}^{\lceil\log f\rceil}T(\vec C(2^j))\right)$ and
  \item uses messages of size at most $M(\vec P) \in O\left(1+\sum_{j=0}^{\lceil\log f\rceil}M(\vec C(2^j))\right)$ bits, 
\end{itemize} 
where the sums are empty for $f=0$.
\end{theorem}
\begin{proof}
We show by induction on $k$ that $f$-resilient strong $\Psi$-pulsers $\vec P(f,\Psi)$ on $n>3f$ nodes with the stated complexity exist for any $f< 2^k$, with the addition that the (bounds on) stabilisation time and message size of our pulsers are non-decreasing in $f$. We anchor the induction at $k=0$, i.e., $f=0$, for which, trivially, a $0$-resilient strong $\Psi$-pulser with $n\in \bN$ nodes is given by one node generating pulses locally and informing the other nodes when to do so. This requires $1$-bit messages and stabilises in $\Psi+1$ rounds.

Now assume that $2^k\leq f < 2^{k+1}$ for $k\in \bN_0$ and the claim holds for all $0 \leq f' < 2^k$. Since $2\cdot (2^k-1)+1 = 2^{k+1}-1$, there are $f_0,f_1<2^k$ such that $f=f_0+f_1+1$. Moreover, as $n>3f>3f_0+3f_1$, we can pick $n_i>3f_i$ for both $i\in \{0,1\}$ satisfying $n=n_0+n_1$. Let $\vec P({f',\Psi'})$ denote a strong $\Psi'$-pulser that exists by the induction hypothesis for $f'<2^k$.

Choose $\Phi \in O(\log \Psi)+T(\vec C(f))$ in accordance with \theoremref{thm:multi-valued} for $L=\Psi$; without loss of generality we may assume that the $O(\log \Psi)$ term is at least 2, that is, $\Phi \ge 2 + T(\vec C(f))$. We apply \theoremref{thm:weak-pulsers} to $\vec C(f)$ and $\vec P_i=\vec P({f_i,\Psi_i})$, where $\Psi_0 = 2\Phi$ and $\Psi_1=3\Phi$, to obtain a weak $\Phi$-pulser $\vec W$ with resilience $f$ on $n$ nodes and stabilisation time of
\begin{align*}
  T(\vec W) &\in \max\{T(\vec P_0),T(\vec P_1) \} + O(\Phi),
\end{align*}
and message size of 
\begin{align*}
  M(\vec W) &\in \max\{M(\vec P_0),M(\vec P_1) \} + O(M(\vec C(f))).
\end{align*}
Next, we apply \theoremref{thm:multi-valued} to $\vec C(f)$ to obtain an $f$-resilient $\Psi$-value consensus protocol $\vec C'$ that uses $M(\vec C(f))$-bit messages and runs in $T(\vec C') \le \Phi$ rounds. We feed the weak pulser $\vec W$ and the multivalued consensus protocol $\vec C'$ into \corollaryref{cor:strong-pulsers} to obtain an $f$-resilient strong $\Psi$-pulser $\vec P$ with 
a stabilisation time of 
\begin{align*}
  T(\vec P) & \le T(\vec C') + T(\vec W)+\Psi \le T(\vec W) + \Psi +\Phi \\
            & \in \max\{T(\vec P_0),T(\vec P_1) \} + \Psi + O(\Phi)
\end{align*}
and message size bounded by
\begin{align*}
  M(\vec P) & \le M(\vec W)+M(\vec C(f)) \\
            & \in  \max\{M(\vec P_0),M(\vec P_1) \} + O(M(\vec C(f))).
\end{align*}

Applying the bounds given by the induction hypothesis to $\vec P_0$ and $\vec P_1$, the definitions of $\Phi$, $\Psi_0$ and $\Psi_1$, and the fact that both $T(\vec C(f))$ and $M(\vec C(f))$ are non-decreasing in $f$, we get that the stabilisation time satisfies
\begin{align*}
  T(\vec P) &\in \max\{T(\vec P({f_0,\Psi_0})),T(\vec P({f_1,\Psi_1})) \} + \Psi + O(\Phi) \\
            &\subseteq (1+o(1))\cdot 3\Phi + O\left(\sum_{j=0}^{\lceil \log 2^k \rceil}T(\vec C(2^j))\right) + \Psi + O(\Phi) \\
            &\subseteq \Psi + O(\log \Psi) + O\left(\sum_{j=0}^{\lceil \log 2^k \rceil}T(\vec C(2^j))\right) + O(T(\vec C(f)))  \\
            &\subseteq (1+o(1))\Psi + O\left(\sum_{j=0}^{\lceil \log f \rceil}T(\vec C(2^j))\right), 
\end{align*}
and message size is bounded by
\begin{align*}
  M(\vec P) & \in  \max\{M(\vec P({f_0,\Psi_0})),M(\vec P({f_1,\Psi_1})) \} + O(M(\vec C(f))) \\
            & \subseteq O\left(1+\sum_{j=0}^{\lceil\log 2^k \rceil}M(\vec C(2^j))\right) + O(M(\vec C(f))) \\
            & \subseteq O\left(1+\sum_{j=0}^{\lceil\log f\rceil}M(\vec C(2^j))\right).
\end{align*}
Because we bounded complexities using $\max_i \{ T(\vec P_i) \}$, $\max \{ M(\vec P_i) \}$, $T(\vec C(f))$ and $M(\vec C(f))$, all of which are non-decreasing in $f$ by assumption, we also maintain that the new bounds on stabilisation time and message size are non-decreasing in $f$. Thus, the induction step succeeds and the proof is complete.
\end{proof}

Plugging in the phase king protocol~\cite{berman89consensus}, which has optimal resilience, running time $O(f)$, and constant message size, we can extract a strong pulser that is optimally resilient, has asymptotically optimal stabilisation time, and message size $O(\log f)$.
\begin{corollary}\label{coro:consensus-to-pulsers}
For any $\Psi,f\in \bN$ and $n>3f$, an $f$-resilient strong $\Psi$-pulser on $n$ nodes with stabilisation time $(1+o(1))\Psi + O(f)$ and message size $O(\log f)$ exists.
\end{corollary}

We obtain efficient solutions to the firing squad and synchronous counting problems.
\begin{corollary}
For any $f\in \bN$ and $n>3f$, an $f$-resilient firing squad on $n$ nodes with stabilisation and response times of $O(f)$ and message size $O(\log f)$ exists.
\end{corollary}
\begin{proof}
We use \corollaryref{coro:consensus-to-pulsers} with $\Psi\in O(f)$ being the running time of the phase king protocol~\cite{berman89consensus}, followed by applying \theoremref{thm:squad} to the obtained pulser and the phase king protocol.
\end{proof}

\begin{corollary}
For any $C,f\in \bN$ and $n>3f$, an $f$-resilient $C$-counter on $n$ nodes with stabilisation time $O(f + \log C)$ and message size $O(\log f)$ exists.
\end{corollary}
\begin{proof}
In the last step of the construction of \theoremref{thm:consensus-to-pulsers}, we do not use \corollaryref{cor:strong-pulsers} to extract a strong pulser, but directly obtain a counter using \theoremref{thm:strong-agreement}. This avoids the overhead of $\Psi$ due to waiting for the next pulse. Recalling that the $o(\Psi)$ term in the complexity comes from the $O(\log \Psi)$ additive overhead in time of the multi-value consensus routine, the claim follows.
\end{proof}

We remark that one can strengthen the bound on the stabilisation time to $O(f+(\log \Psi)/B)$ using messages of size $B$, by using larger messages in the reduction given by \theoremref{thm:multi-valued}~\cite{lenzen13time}. However, this affects the asymtotic stabilisation time only if $\Psi$ is super-exponential in $f$.

\section{Probabilistic sublinear-time algorithms}\label{sec:rand}

So far, we have confined our discussion to the deterministic setting. However, it is straightforward to adapt our framework to also utilise \emph{randomised} consensus routines, which can break the linear-in-$f$ bound for consensus~\cite{feldman89optimal} and attain better bit complexities than deterministic algorithms~\cite{king11breaking}. Indeed, Ben-Or al.~\cite{ben-or08fast} have shown how to obtain randomised counting algorithms that stabilise in $O(1)$ expected time. However, these algorithms rely on a shared coin, which is costly in terms of communication. 

We now use our framework to obtain fast and communication-efficient \emph{probabilistic} pulsers that stabilise in $\polylog f$ communication rounds, where algorithms need to broadcast only $\polylog f$ bits per round. Here, a probabilistic pulser means that after stabilisation the pulser $\vec P$ may fail to behave correctly in round $t \geq T(\vec P)$ with some small positive probability after which it needs to re-stabilise again.

\subsection{Using probabilistic consensus routines}

For our framework, we require that the running time of the underlying consensus algorithms satisfy deterministic running time bounds, while we allow for a probabilistic guarantee on the agreement and validity properties. That is, we need \emph{Monte Carlo} consensus algorithms. 
Accordingly, we demand that the agreement and validity properties of the Monte Carlo consensus algorithm hold with probability $1-p$, where the probability of failure is $p \leq 1/f^c$ for a sufficiently large constant $c$. Noting that our recursive construction of strong pulsers involves $f^{O(1)}$ calls to the utilised consensus routine within $f^{O(1)}$ rounds, it follows from the union bound that with probability at least 
\[
 1 - \sum^{f^{O(1)}}_{i=1} 1/f^c \ge 1 - 1/f^{c-O(1)}
\]
all consensus instances succeed. These observations give the following generalisation of \theoremref{thm:consensus-to-pulsers}.

\begin{theorem}
Suppose that for constant $\varepsilon\geq 0$ we are given a family of $f$-resilient consensus algorithms $\vec C(f)$ running on any number $n>(3+\varepsilon)f$ of nodes in $T(\vec C(f))$ rounds using $M(\vec C(f))$-bit messages, where $T(\vec C(f))$ and $M(\vec C(f))$ are increasing in $f$, and $\vec C(f)$ fails with probability $p \le 1/f^c$ for sufficiently large $c\in O(1)$. Then, for any $\Psi,f\in \bN$ and $n>(3+\varepsilon)f$, a strong probabilistic $\Psi$-pulser $\vec P$ on $n$ nodes with
\begin{itemize}[noitemsep]
  \item $T(\vec P) \in (1+o(1))\Psi + O\left(\sum_{j=0}^{\lceil\log f\rceil}T(\vec C(2^j))\right)$
  \item $M(\vec P) \in O\left(1+\sum_{j=0}^{\lceil\log f\rceil}M(\vec C(2^j))\right)$
\end{itemize} 
exists, where for $f=0$ the sums are empty and on any round $t \ge T(\vec P)$ the algorithm $\vec P$ fails with probability $f^{O(1)}p$ (and then needs to re-stabilise). 
\end{theorem}

The additional reservation that $\vec C$ may require $n>(3+\varepsilon)f$ accounts for the fact that various randomised consensus protocols have slightly suboptimal resilience. Note also that any further model requirements of the randomised consensus protocols, such as private channels, of course still apply when employing our framework.

\subsection{Probabilistic pulsers, counting and firing squads}

As a concrete example, we plug in the consensus algorithm by King and Saia~\cite{king11breaking}, as it satisfies the properties we need. We now make the additional assumptions that (1) the number of faults is restricted to $f<n/(3+\varepsilon)$ (for arbitrarily small constant $\varepsilon>0$) and (2) communication is via private channels, i.e., faulty nodes behavior in round $t$ is a function of all communication from correct nodes to faulty nodes in rounds $t'\leq t$. 

\begin{theorem}[\cite{king11breaking}]
There exists a protocol $\vec C$ that with probability $1-1/f^c$ solves consensus in $\polylog f$ rounds using messages of size $\polylog f$, provided $f<n/(3+\varepsilon)$ and communication is via private channels.
\end{theorem}

We remark that the consensus algorithm from~\cite{king11breaking} actually limits the number of bits sent by each node to $O(\sqrt{n}\polylog n)$, but in our framework each node broadcasts $\Omega(\log f)$ bits per round.

\begin{corollary}
For any $\Psi,f\in \bN$, $n>(3+\varepsilon)f$ and constant $c$, an $f$-resilient strong probabilistic $\Psi$-pulser on $n$ nodes with stabilisation time $(1+o(1))\Psi + \polylog f$ and message size $\polylog f$ exists, where after stabilisation the algorithm will fail on any round with probability at most $p=1/f^c$.
\end{corollary}

Similarly as before, we can obtain efficient probabilistic counting and firing squads algorithms from the probabilistic pulsers.

\begin{corollary}
For any $f\in \bN$, $n>(3+\varepsilon)f$ and constant $c$, an $f$-resilient firing squad on $n$ nodes with stabilisation and response times of $\polylog f$ and message size $\polylog f$ exists, where after stabilisation the algorithm will fail on any round with probability at most $p=1/f^c$.
\end{corollary}

\begin{corollary}
For any $C,f\in \bN$, $n>(3+\varepsilon)f$, and constant $c$, an $f$-resilient $C$-counter on $n$ nodes with stabilisation time $O(\polylog f+\log C)$ and message size $\polylog f$ exists, where after stabilisation the algorithm will fail on any round with probability at most $p=1/f^c$.
\end{corollary}

We note that we choose a failure probability of $1/f^{\Theta(1)}$ for illustrative purposes; by increasing the running time of the underlying consensus routine (incurring the corresponding linear increase in stabilisation time), one can decrease the failure probability exponentially.

\section{Extensions to other fault models}\label{sec:other}

In this section, we utilise our framework under more benign fault models than the one given by Byzantine faults. This allows us to tolerate a larger amount of faulty nodes: for example, while one cannot tolerate more than $f < n/3$ Byzantine faulty nodes, it is possible to tolerate any number of $f < n$ crash faults or $f < n/2$ send omission faults. 

We start by giving a simple and efficient algorithm for synchronous counting under crash faults; here, our framework is overkill, and a direct approach suffices. Together with the approach used in \sectionref{sec:firing-squad} and a crash-tolerant consensus algorithm, we readily obtain an efficient firing squad protocol in the crash fault setting. After this, we illustrate how to modify the construction of strong and weak pulsers given in \sectionref{sec:strong-pulsers} and \sectionref{sec:weak-pulsers} to work with omission faults. This highlights one of the key features of our construction: the resilience of the underlying consensus routine essentially dictates what kind of -- and how many -- permanent faults our self-stabilising counting and firing squad algorithms tolerate, while only making minor modifications to the various voting steps used in the construction.

\subsection{Counting and firing squads under crash faults\label{sec:crash-faults}}

Crash faults are perhaps the most benign fault type: the nodes do not send misinformation and, in the synchronous setting, all nodes can eventually detect which nodes have crashed. Thus, unlike in the Byzantine setting, designing algorithms under crash faults is relatively easy, as nodes crash cleanly and cause no further trouble.

\begin{definition}[Crash faults]
A \emph{crashing} node stops executing the algorithm in some round $r\in \bN$. In this round, the node manages to send only a subset of the messages it would send if it ran correctly. Thus, only a subset of the respective recipients receive a message from the crashed node in this round. The remaining nodes (and, in rounds $r'>r$ all nodes) receive no message.
\end{definition}

The benign nature of crash faults allows us to use more strict requirements in the synchronous counting and firing squad problems as we will see. In the following, let us use $F(t') \subseteq V$ to denote the set of nodes that have crashed before or in round $t'$. 

\paragraph{Optimal crash-tolerant counting.} Let us start with a definition of the synchronous counting problem under crash faults. The problem is defined similarly as in the case of Byzantine faults, but with the requirement that agreement and consistency are satisfied by the set of currently non-crashed nodes.

\begin{definition}[Counting with crash faults]
In synchronous $C$-counting with crash faults, an execution of an algorithm stabilises in round $t \in \bN$ if and only if all $t\leq t'\in \bN$ the output counters $c(\cdot)$ satisfy
\begin{enumerate}[label=SC\arabic*.,noitemsep]
 \item {\bf Agreement:} $c(v,t')=c(w,t')$ for all $v,w \in V \setminus F(t')$ and 
 \item {\bf Consistency:} $c(v,t'+1)=c(v,t')+1\bmod C$ for all $v,w\in V \setminus F(t'+1)$.
\end{enumerate}
\end{definition}

We now give a simple counting algorithm that attains optimal stabilisation time and resilience under crash faults. Let $c(v,t) \in [C]$ be a local variable that indicates the counter value of node $v$ on round $t$. On every round, every node $v$ broadcasts the value $c(v,t)$ to all other nodes. For every $u,v \in V$, let $c(v,u,t) \in [C] \cup \{ * \}$ denote the value node $v$ receives from node $u$ at the start of round $t+1$. Here, we use the special value $*$ to indicate that node $v$ received no message from node $u$. Observe that we have the guarantee that for any non-crashed node $v \in V \setminus F(t+1)$, we have that $c(v,u,t+1)=*$ for all crashed nodes $u \in F(t)$.

Let $U(v,t) = \{ u \in V : c(v,u,t) \neq * \}$ be the set of nodes $v$ received a message from at the start of round $t+1$. Node $v$ updates its counter value on round $t+1$ by picking the majority value among the values it received:
\[
 c(v,t+1) = \begin{cases} 
              x + 1 \bmod C & \text{if } |\{u \in U(v,t) : c(v,u,t)=x \}| > |U(v,t)|/2, \\
              0 & \text{otherwise.}
            \end{cases}
\]

\begin{lemma}\label{lemma:crash-stab}
 Suppose no node crashes on round $t$. Then for any $u,v \in V \setminus F(t')$ and all $t'>t$, we have that $c(u,t')=c(v,t')$ and $c(u,t'+1) = c(u,t') + 1 \bmod C$.
\end{lemma}
\begin{proof}
 Since no node crashes on round $t$, we have that $U(u,t)=U(v,t)$. Hence, both $u$ and $v$ set the same value $x$ for their counter for round $t+1$ when using the above update rule and we have $c(v,t+1)=c(u,t+1)$. It remains to argue that non-crashed nodes will not ever disagree on their counter values after round $t+1$. To this end, suppose all non-crashed nodes agree on the output on some round $t'$, that is, there exists $x \in [C]$ such that for all $v \in V \setminus F(t')$ we have $c(v,t')=x$. Now for any $v \in V \setminus F(t'+1)$ and each $w \in U(v,t')$ it holds that $c(v,w,t')=x$. Thus, by the above update rule, node $v \in V \setminus F(t'+1)$ satisfies $c(v,t'+1)=x+1 \bmod C$. 
\end{proof}

\begin{theorem}
Let $C > 1$ and $f < n$. There exists a synchrous $C$-counter for $n$ nodes that tolerates $f$ crash faults and stabilises in $f+1$ rounds, where each node broadcasts $\lceil \log C \rceil$ bits every round. Moreover, if no node crashes on some round $t < f+1$, then the algorithm stabilises on round $t+1$.
\end{theorem}
\begin{proof}
Since there are at most $f$ crash faults, there exists a round $t < f+1$ such that no node crashes. Applying \lemmaref{lemma:crash-stab} to this round implies that the algorithm stabilises. Since nodes only need to communicate their current counter values every round, a node needs to broadcast at most $\lceil \log C \rceil$ bits every round.
\end{proof}

The above algorithm has exactly optimal stabilisation time: it is known that any $t$-round counting algorithm solves consensus in $t$ rounds~\cite{dolev16synthesis}, but even under crash faults consensus requires $f+1$ rounds~\cite{aguilera99simple}. Moreover, the algorithm is ``early-stabilising'' in the sense that if there is no crash on some round $t$, then the algorithm stabilises on round $t+1$ even if some nodes crash on later rounds $t'>t$. Finally, the message size is optimal in the worst case: if there are no crashes on the first round, then it is necessary for the correct nodes to communicate $\lceil \log C \rceil$ bits to stabilise in one round.

\paragraph{Asymptotically optimal crash-tolerant firing squads.} Let us now consider the firing squad problem under crash faults. Observe that Dolev et al.~\cite{dolev12optimal} give a crash-tolerant firing squad algorithm with \emph{exactly} optimal stabilisation and response time. However, their algorithm uses messages of size $O(f \log f)$ for $f \in \Theta(n)$. We now show that if one relaxes the stabilisation and response times to be \emph{asymptotically} optimal, then messages of size $O(\log f)$ suffice.

\begin{definition}[Firing squad with crash faults]
In the firing squad problem with crash faults, we say that an execution of an algorithm \emph{stabilises in round $t\in \bN$} if the following properties hold:
\begin{itemize}[noitemsep]
  \item {\bf Agreement:} $\sfire(v,t')=\sfire(w,t')$ for all $v,w\in V\setminus F(t')$ and $t\leq t'\in \bN$.
  \item {\bf Safety:} If $\sfire(v,t_F)=1$ for $v\in V\setminus F(t_F)$ and $t\leq t_F\in \bN$, then there is $t_F\geq t_G\in \bN$ such that
 \begin{enumerate}[label=(\roman*)]
  \item $\sgo(w,t_G)=1$ for some $w\in V\setminus F(t_G)$, and 
  \item $\sfire(v,t')=0$ for all $t'\in \{t_G+1,\ldots,t_F-1\}$.
 \end{enumerate}
  \item {\bf Liveness:} If $\sgo(v,t_G)=1$ for $v\in V\setminus F(t_G+1)$ and $t\leq t_G\in \bN$, then $\sfire(v,t_F)=1$ for all nodes $v\in V\setminus F(t_F)$ and some $t_G< t_F\in \bN$.
\end{itemize}
\end{definition}

In \sectionref{sec:firing-squad} we saw that firing squad can be solved easily using consensus and a strong pulser algorithm. The same reduction works also under crash faults. The only difference is that we need to modify the second line of the firing squad algorithm given in \sectionref{ssec:counting-squad}. We replace the condition of seeing at least $f+1$ times $\sgo(w,t-1)=1$ with seeing at least \emph{one} node $w$ with $\sgo(w,t-1)=1$. This yields an result analogous to \theoremref{thm:squad} under crash faults.

Similarly, in the case of consensus under crash faults, the agreement and validity conditions need to be satisfied by all non-crashed nodes at the end of the execution. In this setting, consensus can be solved in $f+1$ rounds using $1$-bit messages~\cite{raynal10survey}. For example, we can adapt the same majority voting technique as in the counting algorithm above for $f+1$ rounds to solve consensus as well. For $C=f+1$, we can use a crash-tolerant $C$-counter to obtain a crash-tolerant strong $C$-pulser using \lemmaref{lemma:strong_pulsers_equal_counters}. Using similar arguments as in \theoremref{thm:squad}, we obtain the following result.

\begin{corollary}
For any $f\in \bN$ and $n>f$, there exists an $f$-crash-tolerant firing squad on $n$ nodes with stabilisation and response times of $O(f)$ and message size $O(\log f)$.
\end{corollary}

\subsection{The framework under omission faults}

We consider a fault type that falls between crash and Byzantine faults: omission faults. The case of omission faults is more challenging than crash faults, as faulty nodes may drop some of the messages, while still continuing to participate in the execution of the algorithm for indefinitely long. For simplicity, we focus on send omission faults, as our primary goal here is to demonstrate the flexibility our framework. One could also consider e.g.\ receive or general omission faults~\cite{raynal10survey}.

\begin{definition}[Omission faults]
We say that a node $v\in V$ suffers from (send) \emph{omission} faults if in each round $r$ the messages sent by $v$ are only received by some (arbitary) subset $U(r) \subseteq V$ of the nodes only. The remaining nodes in $V \setminus U(r)$ receive no message from $v$.
\end{definition}

Note that under send omission faults, the faulty nodes still receive messages from correct nodes. Hence, we modify the definitions of synchronous counting and firing squad problems as follows.

\begin{definition}[Counting with omission faults]
In the synchronous $C$-counting problem with omission faults, we require that the agreement and consistency conditions are satisfied by all nodes.
\end{definition}

\begin{definition}[Firing squad with omission faults]
In the firing squad problem with omission faults, the agreement, safety, and liveness conditions are adapted as follows. We say that an execution of an algorithm \emph{stabilises in round $t\in \bN$} if the following three properties hold:
\begin{itemize}[noitemsep]
  \item {\bf Agreement:} $\sfire(v,t')=\sfire(w,t')$ for all $v,w\in V$ and $t\leq t'\in \bN$.
  \item {\bf Safety:} If $\sfire(v,t_F)=1$ for $v\in V$ and $t\leq t_F\in \bN$, then there is $t_F\geq t_G\in \bN$ such that 
\begin{enumerate}[label=(\roman*)]
 \item $\sgo(w,t_G)=1$ for some $w\in V$,
 \item $\sfire(v,t')=0$ for all $t'\in \{t_G+1,\ldots,t_F-1\}$.
\end{enumerate}
  \item {\bf Liveness:} If $\sgo(v,t_G)=1$ for $v\in V\setminus F$ and $t\leq t_G\in \bN$, then $\sfire(v,t_F)=1$ for all nodes $v\in V$ and some $t_G< t_F\in \bN$.
\end{itemize}
\end{definition}

Finally, we remark that also the definition of consensus needs to be adapted in the case of send omission faults. For send omission faults, termination, agreement, and validity apply to all nodes in the system. 

\paragraph{Adjustments to the basic framework.} As Byzantine faults also cover omission faults, our framework could be used as-is with minimal modifications. However, weaker fault types permit a larger number of faults to be tolerated. Moreover, we can readily employ consensus protocols tailored for various different fault types from the literature by slightly adapting the voting schemes used in our constructions outside the consensus routines. In addition, we must adapt our reductions of multivalue consensus and silent consensus to standard binary consensus. More precisely, for each fault type, we need to address and handle the following issues:
\begin{enumerate}[noitemsep]
  \item We need an equivalent of \theoremref{thm:multi-valued}.
  \item Distributing knowledge of $\sgo$ inputs of $1$ in the firing squad algorithm.
  \item Adapting the filtering construction.
  \item We need an equivalent of \theoremref{thm:silent-consensus}.
  \item Adapting the pruning algorithm.
\end{enumerate}

In the following, we illustrate how to do the above modifications in the case of send omission faults. To this end, we require that $n>2f$, which is necessary and sufficient to achieve consensus in the presence of omission faults~\cite{raynal10survey}. We handle each of the above points as follows.

\begin{enumerate}
  \item Each node broadcasts its input bit by bit. There is a unique input $x$ that can be received $n-f>f$ times by any node (the threshold must be met for every bit, but the senders may differ). If $v$ receives such an input, it stores it and sends it again bit by bit; if not, it sends nothing in this second transmission. If $x$ is receveived at least $n-f$ times by $v\in V$ in this second iteration, $v$ uses input $1$ in a call to the binary consensus routine, otherwise $0$. If $v$ received \emph{any} value $x$ in this second transmission, it returns it in case the consensus routine outputs $1$. If the routine outputs $0$, it returns $0$. Note that if any node used input $1$, it received $x$ $f+1$ times in the second iteration, entailing that every node received $x$. Thus, agreement holds by the properties of the binary consensus routine. Likewise, validity of the latter implies validity of the former: if all nodes have the same input $x$, it is received $n-f$ times by each node in both iterations.

  \item Again, we replace the threshold of receiving $\sgo(w,t-1)=1$ from $f+1$ nodes $w\in V$ with the threshold of receiving $\sgo(w,t-1)=1$ from any node $w\in V$ in Step 2 of the firing squad algorithm, and adjusting the proof of \theoremref{thm:squad} is straightforward.

  \item In the filtering construction, we replace the requirement from $f<n/3$ to $f<n/2$. The only change is that $\ell_i(v,t+1)$ is set to $0$ if there is any node $u$ sending $m_i(u,t)=1$. One can readily check that this does not affect the correctness of \lemmaref{lemma:one_correct}, \lemmaref{lemma:correct-block-votes}, or \lemmaref{lemma:correct-block-fires}. Concerning \lemmaref{lemma:any-block-fires}, observe that any node having $b_i(v,t)=1$ implies $M_i(v,t)=1$ and thus $m_i(w,t-1)=1$ for at least $n-f>f$ nodes $w\in V$. Hence, each node $u\in V$ receives $m_i(w,t-1)=1$ from at least one node $w\in V$ and sets $l_i(u,t)=0$. \lemmaref{lemma:any-block-fires} now follows by similar reasoning as in the Byzantine case.

  \item We follow the same strategy as for the Byzantine case. In the first two rounds, a node sets its input to $0$ if receiving fewer than $n-f$ times $1$. Any node receiving a message in the first round participates in the execution of the (non-silent) binary consensus protocol. Each node returns $0$ if it received no message in the second round, it was forced to abort the binary consensus protocol due to violation of message size bound or an otherwise invalid execution, or the binary consensus protocol returned $0$. If a node does not participate, there are at most $f$ nodes with non-zero input, implying that no node receives a message in the second round. Thus, agreement holds in this case. If a node uses input one for the call to the non-silent consensus routine, all nodes participate, as at least $f+1$ nodes sent $1$ in the first round. Thus agreement follows from the correct execution of the non-silent protocol. Silence and validity are easily verified.

  \item We modify Step 1 of the pruning algorithm to set $r_i(v,t+1)=1$ if received $b_i(w,t)=1$ from any $w\in V$. It follows that if any node $v\in V$ uses input $1$ for a consensus instance whose first round is simulated in round $t$, each node received $b_i(v,t-1)=1$ and thus participates in the instance. Moreover, if all nodes have $b_i(v,t-1)=1$, all use input $1$ for the instance. Similar reasoning to the Byzantine case now establishes the required properties of the pruning routine.
\end{enumerate}

\paragraph{Results for omission faults.}
None of the above modifications change message size or time bounds, implying that we can feed the modified machinery with an arbitrary binary consensus algorithm resilient to $f<n/2$ omission faults to obtain results analogous to the Byzantine case.

\begin{theorem}
Suppose that we are given a family of $f$-omission-resilient deterministic consensus algorithms $\vec C(f)$ running on any number $n>2f$ of nodes in $T(\vec C(f))$ rounds using $M(\vec C(f))$-bit messages, where $T(\vec C(f))$ and $M(\vec C(f))$ are non-decreasing in $f$. Then, for any $\Psi,f\in \bN$ and $n>2f$, a strong $f$-omission-resilient $\Psi$-pulser $\vec P$ on $n$ nodes with
\begin{itemize}[noitemsep]
  \item $T(\vec P)\in (1+o(1))\Psi + O\left(\sum_{j=0}^{\lceil\log f\rceil}T(\vec C(2^j))\right)$
  \item $M(\vec P)\in O\left(1+\sum_{j=0}^{\lceil\log f\rceil}M(\vec C(2^j))\right)$
\end{itemize} 
exists, where for $f=0$ the sums are empty.
\end{theorem}

Plugging in a folklore algorithm for $f<n/2$ faults (e.g. the first omission-resilient algorithm described in~\cite{raynal10survey}) or applying the phase king algorithm under omission faults, we obtain the following results.

\begin{corollary}
For any $\Psi,f\in \bN$ and $n>2f$, there exists a deterministic $f$-omission-resilient strong $\Psi$-pulser on $n$ nodes with stabilisation time $(1+o(1))\Psi + O(f)$ and message size of $O(\log f)$~bits.
\end{corollary}

\begin{corollary}
For any $f\in \bN$ and $n>f$, there exists a deterministic $f$-omission-resilient firing squad on $n$ nodes with stabilisation and response times of $O(f)$ and message size of $O(\log f)$~bits.
\end{corollary}

\begin{corollary}
For any $C,f\in \bN$ and $n>3f$, there exists a deterministic $f$-omission-resilient $C$-counter on $n$ nodes with stabilisation time $O(f + \log C)$ and message size of $O(\log f)$ bits.
\end{corollary}

\section*{Acknowledgements}

We are grateful to Danny Dolev for inspiring discussions and valuable comments, especially concerning silent consensus. We thank anonymous reviewers for their comments on an earlier draft of this manuscript.


\DeclareUrlCommand{\Doi}{\urlstyle{same}}
\renewcommand{\doi}[1]{\href{http://dx.doi.org/#1}{\footnotesize\sf
    doi:\Doi{#1}}}
\bibliographystyle{plainnat}
\bibliography{squad}

\end{document}